\documentclass[10pt,a4paper]{article}
\usepackage{authblk}
\usepackage[colorlinks,linkcolor=blue,anchorcolor=red,citecolor=red,pdfstartview=FitH]{hyperref}
\usepackage{graphics,graphicx}
\usepackage{mathtools}
\usepackage{amsmath,amssymb,amsthm,accents,amsfonts}
\usepackage{yhmath}
\usepackage{color}
\usepackage{cite}
\usepackage{dsfont}
\usepackage{bm}
\usepackage{rotating}
\usepackage{type1cm}
\usepackage{float}
\usepackage{tikz}
\usetikzlibrary{arrows,chains,matrix,positioning,scopes}
\makeatletter
\tikzset{join/.code=\tikzset{after node path={%
\ifx\tikzchainprevious\pgfutil@empty\else(\tikzchainprevious)%
edge[every join]#1(\tikzchaincurrent)\fi}}}
\makeatother
\tikzset{>=stealth',every on chain/.append style={join},
         every join/.style={->}}
\tikzstyle{labeled}=[execute at begin node=$\scriptstyle,
   execute at end node=$]

\def \b#1{\overline{#1}}
\def \t#1{\widetilde{#1}}
\def \h#1{\widehat{#1}}

\newtheorem{prop}{Proposition}[section]

\numberwithin{equation}{section}
\allowdisplaybreaks

\begin{document}

\title{Binary Darboux transformations for discrete modified Boussinesq equation}
\author[1]{Ying Shi\footnote{E-mail:yingshi@zust.edu.cn}}
\author[2]{Junxiao Zhao}
\affil[1]{\normalsize School of Science, Zhejiang University of Science and Technology, Hangzhou 310023, P.R. China}
\affil[2]{\normalsize School of Mathematics and Statistics, University of Glasgow, University Place, Glasgow G12 8QQ}

\date{}
\maketitle

\begin{abstract}
We obtain the well-known discrete modified Boussinesq equation in two-component form as well as its Lax pair in $3\times3$ matrix form through a 3-periodic reduction technique on the Hirota-Miwa equation and its Lax pair. We describe how Darboux transformations and binary Darboux transformations can be constructed for this two-component discrete integrable equation. These transformations are then used to obtain classes of explicit solutions in the form of Casorati- and Gramm-type determinants. $N$-soliton solutions of the discrete modified Boussinesq equation are discussed as well when taking the vacuum potentials as constants.
\end{abstract}

\section{Introduction}
Due to rich algebraic structure and crucial contributions to other fields of modern mathematics and physics, discrete integrable equations in a sense are playing increasingly important role in the theory of integrable systems, see e.g.\cite{HJN-2016}. The property named multidimensional consistency is one of its key feature \cite{NS-1997, DS-1997, NW-2001}. This property was used efficiently as a classifying property \cite{ABS-2003}.

Multi-component discrete integrable equations form an important class of the discrete integrable systems. But the work on the multi-component lattice equations \cite{NP-1992, MK-2010, XN-2011, H-2011, ZZN-2012} is still rather sporadic when compared with that of scalar lattice equations. A classic example of multi-component integrable lattice equations is the lattice Gel'fand-Dikii hierarchy \cite{NP-1992}. In this lattice hierarchy, the first higher-rank case is the discrete Boussinesq equation in three-component form on a quadrilateral or equivalently in 9-point scalar form on a $3\times3$ stencil on the discrete two-dimensional plane. By using the multidimensional consistency, some results on the classification of the Boussinesq-type multi-component lattice equations are given in \cite{H-2011} (see also \cite{ZZN-2012} from the viewpoint of the direct linearization approach). Furthermore, in a recent paper \cite{FX-2017}, another algebraic classification scheme of discrete integrable systems (including the multi-component case) is also shown.

When compared with scalar lattice equations, approach of searching exact solutions of multi-component discrete integrable equations is still not so rich. In recent studies, the Hirota bilinear approach (see \cite{HZ-2010}) and a reduction technique (see \cite{MK-2010}) both work well in dealing with finding soliton solutions of the multi-component discrete Boussinesq-type equations ( e.g., the discrete Boussinesq (three-component) and modified Boussinesq (two-component) equations). Moreover, another systematic and effective approach is the direct linearization method \cite{NP-1992, ZZN-2012}. But this approach is difficult to be mastered, since it employs an integral equation (not easy to construct) and the knowledge of infinite-component matrix theory. Therefore, our motivation is to find an algebraic and straightforward  method for solving the multi-component discrete integrable equations.

Darboux transformations have been proven to be an powerful and straightforward approach for constructing exact solutions of (both continuous and discrete) integrable systems\cite{M-1991, WLG-1997, N-1997}. Moreover, this approach has close connections with other aspects of the study of integrable (infinite-dimensional) nonlinear systems, e.g., the Sato theory within the language of free fermion operators \cite{WTLS-1998}. The main idea of this approach is to find appropriate gauge transformations keeping the linear problems associated with integrable systems being invariant, so that the relationships between the new and the original potentials (eigenfunctions as well) can be built. The advantage of Darboux transformations and binary Darboux transformations lies in not only that the solutions are fruitful (not just soliton solutions) and expressible in terms of determinants (the Wronskian, Casoratian and Grammian determinants), but also that the iterative algorithms in them are purely algebraic.
The aim of this paper, therefore, is to investigate this approach, especially the binary Darboux transformations, for solving the multi-component discrete integrable equations (e.g., the two-component discrete modified Boussinesq equation).  Binary Darboux transformation can be seen as a compound type of the Darboux transformations. It involves two linear problems associated with the same integrable systems.  The advantage of binary Darboux transformation is that it gives rise to solutions in the form of Grammian-determinants type. In recent papers \cite{SNZ-2014, SNZ-2017}, the Darboux transformations and binary Darboux transformations in the light of some reduction technique was applied effectively to the discrete integrable equations.  This series work enables us to study the Darboux and binary Darboux transformations for the multi-component discrete integrable equations from this viewpoint in the current paper.

This paper is organized as follows. In Section 2, we explain a 3-periodic reduction technique on the Hirota-Miwa equation and its Lax pair.  It involves some simple 3-periodic constraints that can be imposed on the $\tau$-functions in the Hirota-Miwa equation (a bilinear equation).  The two-component discrete modified Boussinesq equation and its Lax pair in matrix form are obtained under this reduction technique (corresponding to the results in \cite{FX-2017}, i.e., expressions (3.8a) and (3.8b), by a classification scheme).  The reduction technique here sheds new light on the constructions of Darboux transformations and binary Darboux transformations. In Section 3, through the 3-periodic reduction conditions, we investigate the Darboux transformations of the discrete modified Boussinesq equation. Particularly, we find out that the 3-periodic property is preserved under the Darboux transformations. In Section 4, we describe in detail how binary Darboux transformations can be constructed for the discrete modified Boussinesq equation. In Section 5,  by choosing the seed solution of the discrete modified Boussinesq equation as constants $(1, 1)$, we present explicit examples of the classes of solutions that may be obtained by means of the Darboux transformations and binary Darboux transformations.

\section{From the Hirota-Miwa equation to the discrete modified Boussinesq equation: 3-periodic reduction}\label{sec-2}

The Hirota-Miwa equation \cite{H-1981, M-1982} is the three-dimensional discrete integrable system
\begin{equation}\label{H-M-1}
(a_1-a_2)\tau_{12}\tau_3+(a_2-a_3)\tau_{23}\tau_1+(a_3-a_1)\tau_{31}\tau_2=0,
\end{equation}
where lattice parameters $a_k$ are constants, $k=1,2,3$ and $\tau=\tau(n_1,n_2,n_3)$.  The notations we adopted here and later are as follows, within forward shift operators $T_{n_1}$, $T_{n_2}$ and $T_{n_3}$:
\begin{align*}
\tau_1:=T_{n_1}(\tau)=\tau(n_1+1, n_2,n_3), ~\tau_2:=T_{n_2}(\tau)=\tau(n_1, n_2+1,n_3), ~\tau_3:=T_{n_3}(\tau)=\tau(n_1, n_2,n_3+1).
\end{align*}
The compositions of the shift operators $T_{n_1}$, $T_{n_2}$ and $T_{n_3}$ are denoted by the combinations of the subscripts 1, 2 and 3, e.g., $\tau_{12}:=T_{n_1}T_{n_2}(\tau)=\tau(n_1+1, n_2+1, n_3)$.

The Hirota-Miwa equation \eqref{H-M-1} arises as the compatibility conditions of the linear system
\begin{equation}\label{H-M-LP-1}
\phi_i-\phi_j=(a_i-a_j)u^{ij}\phi, \quad  \mathrm{with}\quad u^{ij}=\frac{\tau_{ij}\tau}{\tau_i\tau_j},\quad (1\leq i<j\leq3),
\end{equation}
where for $\phi=\phi(n_1,n_2,n_3)$ each subscript $i$ denotes a forward shift in the corresponding discrete variable $n_i$. For example, $\phi_{1}\!=\!T_{_{n_1}}\!(\phi)\!=\!\phi(n_1+1,n_2,n_3)$.

For the Hirota-Miwa equation \eqref{H-M-1}, in  \cite{N-1997,SNZ-2014}, its one soliton solution in discrete exponential function form is expressed as follows
\begin{equation}\label{tau-solution}
  \tau(n_1,n_2,n_3)=1+\alpha\prod_{i=1}^{3}\left(\frac{a_i+p^\prime}{a_i+p}\right)^{n_i}.
\end{equation}
Introduce $f(n_1,n_2,n_3)$, $g(n_1,n_2,n_3)$ and $h(n_1,n_2,n_3)$ through the following  transformations
\begin{equation}\label{f-g-h-condition}
 \tau=f=T_{n_3}^3(f),\quad h=T_{n_3}^2(f),\quad g=T_{n_3}(f).
\end{equation}
Note here that the reduction condition \eqref{f-g-h-condition} gives  $a_3=0$, $p^\prime=\omega p$ with $\omega^3=1$ for $\omega\neq1$, in \eqref{tau-solution}, and
\begin{subequations}\label{g-h-condition}
\begin{align}
g=T_{n_3}^3(g),\quad f=T_{n_3}^2(g),\quad h=T_{n_3}(g),\\
h=T_{n_3}^3(h),\quad g=T_{n_3}^2(h),\quad f=T_{n_3}(h),
\end{align}
\end{subequations}
where it is easy to find the cyclical property of $f \!\to\!g\!\to\!h\!\to\!f$  with respect to $n_3$.  By applying the reduction condition \eqref{f-g-h-condition} to the Hirota-Miwa equation \eqref{H-M-1}, together with parameter reduction $a_3=0$, we get
\begin{subequations}\label{H-M-2}
\begin{align}
(a_1-a_2)f_{12}g+a_2g_2f_1-a_1g_1f_2=0,\label{Bilinear-f-g-h-1}\\
(a_1-a_2)g_{12}h+a_2h_2g_1-a_1h_1g_2=0,\label{Bilinear-f-g-h-2}\\
(a_1-a_2)h_{12}f+a_2f_2h_1-a_1f_1h_2=0.\label{Bilinear-f-g-h-3}
\end{align}
\end{subequations}
The way to get the equations \eqref{Bilinear-f-g-h-2} and \eqref{Bilinear-f-g-h-3} is applying the cyclical property $f \!\to\!g\!\to\!h\!\to\!f$ to the equation \eqref{Bilinear-f-g-h-1}, or equivalently taking shift operators $T_{n_3}$ and $T_{n_3}^2$ to act on \eqref{Bilinear-f-g-h-1} and using the reduction condition \eqref{f-g-h-condition} and \eqref{g-h-condition}.

Define three functions (potentials)
\begin{equation}\label{potential}
  v(n_1,n_2,n_3)=\frac{h}{f},\quad w(n_1, n_2, n_3)=\frac{g}{f},\quad u(n_1,n_2,n_3)=\frac{f_{12}f}{f_1f_{2}}.
\end{equation}
By substituting \eqref{potential} into \eqref{H-M-2}, we get
\begin{subequations}\label{d-Miura}
\begin{align}
(a_1-a_2)uw&=a_1w_{1}-a_2w_2, \label{d-Miura-1}\\
(a_1-a_2)uw_{12}v&=a_1w_2v_1-a_2w_{1}v_2, \label{d-Miura-2}\\
(a_1-a_2)uv_{12}&=a_1v_2-a_2v_{1}. \label{d-Miura-3}
\end{align}
\end{subequations}
Eliminating $u$ in \eqref{d-Miura} gives
\begin{equation}\label{d-p-MBSQ}
\frac{a_1w_{1}-a_2w_2}{w}=\frac{a_1w_2v_1-a_2w_{1}v_2}{w_{12}v}=\frac{a_1v_2-a_2v_{1}}{v_{12}},
\end{equation}
which is the well-known two-component discrete modified Boussinesq equation given in \cite{N-1997} by the direct linearization method (see also in \cite{FX-2017}, i.e., expressions (3.8a) and (3.8b), by an algebraic approach). It is also given in \cite{H-2011} as an special case of the C-3 equation.  Moreover,  \eqref{H-M-2} can also be seen as a bilinearization of the discrete modified Boussinesq equation \eqref{d-p-MBSQ}.

With the 3-periodic property of $f$, $g$ and $h$, we have the formulae on the potentials $v$ and $w$ as follows
\begin{align*}
T_{n_3}^3(v)=v,~T_{n_3}^2(v)=wv^{-1}, ~T_{n_3}(v)=w^{-1},\\
T_{n_3}^3(w)=w,~T_{n_3}^2(w)=v^{-1}, ~T_{n_3}(w)=vw^{-1}.
\end{align*}
So the potentials $v$ and $w$ also satisfy the 3-periodic property in the virtual variable $n_3$.

Next, starting by the linear system of the Hirota-Miwa equation \eqref{H-M-LP-1}, we show the way of discovering the linear system in matrix form of the discrete modified Boussinesq equation \eqref{d-p-MBSQ} through a $3$-periodic reduction technique.

Introduce eigenfunctions  $\phi(n_1,n_2,n_3)$, $\b \phi(n_1, n_2, n_3)$ and $\b{\b{\phi}}(n_1, n_2, n_3)$,  and impose a 3-periodic condition  as below
\begin{equation}\label{phi-condition}
\phi=\lambda^{-3}T_{n_3}^3(\phi),\quad \b{\b\phi}=\lambda^{-2} T_{n_3}^2(\phi), \quad \b{\phi}=\lambda^{-1} T_{n_3}(\phi),
\end{equation}
on the eigenfunction $\phi(n_1,n_2,n_3)$ in the linear system \eqref{H-M-LP-1},  where the parameter $\lambda$ serves as the spectral parameter.
From \eqref{phi-condition}, we have
\begin{subequations}\label{bphi-condition}
\begin{align}
\b\phi=\lambda^{-3}T_{n_3}^3(\b\phi),\quad \phi=\lambda^{-2} T_{n_3}^2(\b\phi), \quad \b{\b\phi}=\lambda^{-1} T_{n_3}(\b\phi),\\
\b{\b\phi}=\lambda^{-3}T_{n_3}^3(\b{\b\phi}),\quad \b\phi=\lambda^{-2} T_{n_3}^2(\b{\b\phi}), \quad \phi=\lambda^{-1} T_{n_3}(\b{\b\phi}),
\end{align}
\end{subequations}
where  it means the cyclical property between $\phi \!\to\!\b{\phi}\!\to\!\b{ \b\phi}\!\to\!\phi$ with respect to $n_3$.

By applying the reduction conditions \eqref{f-g-h-condition}, \eqref{g-h-condition}, \eqref{phi-condition} and \eqref{bphi-condition}, together with $a_3=0$ and the potentials' transformation \eqref{potential}, to the linear system \eqref{H-M-LP-1}, we get
\begin{subequations}\label{H-M-LP2}
\begin{align}
\phi_{1}-\phi_{2}&=(a_1-a_2)u\phi,\label{H-M-LP2-1}\\
\phi_{2}-\lambda\b\phi&=a_2w^{-1}w_2\phi, \label{H-M-LP2-2}\\
\phi_{1}-\lambda\b\phi&=a_1w^{-1}w_1\phi.\label{H-M-LP2-3}
\end{align}
\end{subequations}
Then by using the cyclical property of $f \!\to\!g\!\to\!h\!\to\!f$ in \eqref{f-g-h-condition} and \eqref{g-h-condition},  and $\phi \!\to\!\b{\phi}\!\to\!\b{\b\phi}\!\to\!\phi$ in \eqref{phi-condition} and \eqref{bphi-condition}, we get
\begin{subequations}\label{H-M-LP3}
\begin{align}
\b\phi_{1}-\b\phi_{2}&=(a_1-a_2)w_{12}ww^{-1}_1w^{-1}_2u\b\phi,\label{H-M-LP3-1}\\
\b\phi_{2}-\lambda\b{\b\phi}&=a_2v_2ww^{-1}_2v^{-1}\b\phi, \label{H-M-LP3-2}\\
\b\phi_{1}-\lambda\b{\b\phi}&=a_1v_1ww^{-1}_1v^{-1}\b\phi;\label{H-M-LP3-3}
\end{align}
\end{subequations}
and
\begin{subequations}\label{H-M-LP4}
\begin{align}
\b{\b\phi}_{1}-\b{\b\phi}_{2}&=(a_1-a_2)v_{12}vv^{-1}_1v^{-1}_2u\b{\b\phi},\label{H-M-LP4-1}\\
\b{\b\phi}_{2}-\lambda\phi&=a_2vv^{-1}_2\b{\b\phi}, \label{H-M-LP4-2}\\
\b{\b\phi}_{1}-\lambda\phi&=a_1vv^{-1}_1\b{\b\phi}.\label{H-M-LP4-3}
\end{align}
\end{subequations}
Through the expressions in \eqref{d-Miura}, the equations \eqref{H-M-LP2-1}, \eqref{H-M-LP3-1}  and \eqref{H-M-LP4-1} can be derived by \eqref{H-M-LP2-2} and \eqref{H-M-LP2-3}, \eqref{H-M-LP3-2} and \eqref{H-M-LP3-3}, and  \eqref{H-M-LP4-2} and \eqref{H-M-LP4-3} respectively.

Define a vector eigenfunction $\bm\Phi=(\phi,\b\phi, \b{\b\phi})^T$, then the linear equations \eqref{H-M-LP2-3}, \eqref{H-M-LP3-3}, \eqref{H-M-LP4-3} and \eqref{H-M-LP2-2}, \eqref{H-M-LP3-2}, \eqref{H-M-LP4-2} can be  respectively written in matrix form as below
\begin{subequations}\label{matrix-LP}
\begin{align}
&\bm\Phi_1=\bm L\bm\Phi,\label{matrix-LP-1}\\
&\bm\Phi_2=\bm M\bm\Phi, \label{matrix-LP-2}
\end{align}
where
\begin{gather}
\bm L=\left(
\begin{array}{ccc}
 a_1 w_1 w^{-1} & \lambda & 0 \\
0 & a_1 v_1 v^{-1}w w_1^{-1} & \lambda \\
\lambda & 0 & a_1 v v_1^{-1} \\
\end{array}
\right)
,~~~~
\bm M=\left(
\begin{array}{ccc}
 a_2 w_2 w^{-1} & \lambda & 0 \\
0 & a_2 vv^{-1}_2 ww_2^{-1}& \lambda \\
\lambda & 0 & a_2 v v_2^{-1} \\
\end{array}
\right)
.\nonumber
\end{gather}
\end{subequations}
One then finds that
\begin{align*}
0&=\bm\Phi_{12}-\bm\Phi_{21}=(\bm L_2\bm M-\bm M_1\bm L)\bm\Phi\nonumber\\
&=\!\lambda\!\!\left(
    \begin{array}{ccc}
      0&\frac{(a_1\!w_1-a_2\!w_2)\!w_{12}v-(a_1\!w_2v_1-a_2\!w_1\!v_2)\!w}{w_1w_2v}&0\\
      0& 0&\frac{(a_1\!w_2\!v_1-a_2\!w_1\!v_2)v_{12}-(a_1\!v_2-a_2\!v_1)w_{12}v}{w_{12}v_1v_2}\\
      \frac{(a_1\!v_2-a_2\!v_{1})w-(a_1\!w_1-a_2\!w_2)v_{12}}{wv_{12}}&0&0 \\
    \end{array}
  \right)\!\!\bm\Phi.
\end{align*}
So the compatibility condition of the above linear system \eqref{matrix-LP} in eigenfunction $\bm\Phi$ is that $v$ and $w$ obey the discrete modified Boussinesq equation \eqref{d-p-MBSQ}.

\noindent\emph{\textbf{Remark.}} The first-order linear system in matrix form \eqref{matrix-LP} can be equivalently written in third-order linear system in scalar form as follows
\begin{subequations}
\begin{align*}
\phi_{111}-a_1(vw_{11}+v_{11}w_1+v_1w_{111})v_1^{-1}w_{11}^{-1}\phi_{11}+a_1^2(vw_{11}+v_{11}w_1+v_1w)v_1^{-1}w_1^{-1}\phi_1-a_1^3\phi=\lambda^3\phi,\\
\phi_{222}-a_2(vw_{22}+v_{22}w_2+v_2w_{222})v_2^{-1}w_{22}^{-1}\phi_{22}+a_2^2(vw_{22}+v_{22}w_2+v_2w)v_2^{-1}w_2^{-1}\phi_2-a_2^3\phi=\lambda^3\phi.
\end{align*}
\end{subequations}
\section{Darboux transformations}\label{sec-3}

We will now derive the Darboux transformations for the linear system \eqref{matrix-LP} which are compatible ($\bm\Phi_{12}=\bm\Phi_{21}$), if and only if the discrete modified Boussinesq equation \eqref{d-p-MBSQ} is satisfied. To do this we need to define the Casoratian determinant.  Let $\bm\theta=({\theta^{1}}(n_1,n_2,n_3), {\theta^{2}}(n_1,n_2,n_3), \dots, {\theta^{N}}(n_1,n_2,n_3))^T$ be an $N$-vector solution of the linear system \eqref{matrix-LP}.  The Casoratian determinant (with forward-shifts)
\begin{align*}
C(\theta^{1},\theta^{2}, \dots, \theta^{N})\!=\!|\bm\theta,T_{_{n_i}}(\bm\theta), \dots, T_{_{n_i}}^k(\bm\theta),\dots, T_{_{n_i}}^{N-1}(\bm\theta)| ~~~~(1\leq i \leq 3),
\end{align*}
may also be unambiguously defined in the following notation as
\begin{align*}
C_{_{[i]}}(\theta^{1},\theta^{2}, \dots, \theta^{N})\!=\!|\bm\theta(0),\bm\theta(1), \dots, \bm\theta(k), \dots, \bm\theta(N-1)| ~~~~(1\leq i \leq 3),
\end{align*}
where $\bm\theta(k)$ denotes the $N$-vector $({\theta^{1}}(n_1,n_2,n_3), {\theta^{2}}(n_1,n_2,n_3), \dots, {\theta^{N}}(n_1,n_2,n_3))^T$ is subject to the shift $n_i \rightarrow n_i+k$,  $0\leq k \leq N-1$, where $i=1, 2$ or $3$, the same value being taken for $i$ in each column in the determinant.

Let $(v, w)$ be a solution of the discrete modified Boussinesq equation \eqref{d-p-MBSQ} and $\bm\Phi=(\phi, \b\phi, \b{\b\phi})^T$ be a  vector solution of its Lax pair \eqref{matrix-LP}.  The basic Darboux transformation of the equation \eqref{d-p-MBSQ} is given as below.
\begin{prop}\label{prop1-dpmBSQ}
Let $\bm\Theta=(\theta,\b\theta, \b{\b\theta})^T$ having the 3-periodic property  $\theta=\mu^{-3}T_{_{n_3}}^3(\theta), \b{\b\theta}=\mu^{-2}T_{n_3}^2(\theta), \b\theta=\mu^{-1}T_{_{n_3}}(\theta)$ be a non-zero vector solution of the linear system \eqref{matrix-LP} by taking $\lambda=\mu$ for some $(v, w)$, then
\begin{subequations}\label{dpmBSQ-DT-1-1}
\begin{align}
\mathrm {DT}^{\bm\Theta}:
&~~\phi\rightarrow\t\phi=\frac{C_{_{[3]}}(\theta, \phi)} {\theta},~~~~\b\phi\rightarrow\t{\b\phi}=\frac{C_{_{[3]}}(\b\theta, \b\phi)} {\b\theta},~~~~\b{\b\phi}\rightarrow\t{\b{\b\phi}}=\frac{C_{_{[3]}}(\b{\b\theta}, \b{\b\phi})} {\b{\b\theta}},\label{dpmBSQ-DT-1-1-1}\\
&~~v\rightarrow \t v=\frac{T_{_{n_3}}^2\!(\theta)}{\theta}v
\quad \quad w\rightarrow \t w=\frac{T_{_{n_3}}\!(\theta)}{\theta}w
\end{align}
\end{subequations}
leaves \eqref{matrix-LP} invariant.
\end{prop}
The proof of \textbf{Prop.\ref{prop1-dpmBSQ}} is a straightforward computation.  The reader can see that of \textbf{Prop.\ref{prop1N-dpmBSQ}} when $N=1$.  So we do not give details here.

Through the 3-periodic conditions of $(\theta, \b\theta, \b{\b\theta})$ and $(\phi, \b\phi, \b{\b\phi})$, we can obtain the 3-periodic relationships for the new eigenfunctions $(\t\phi, \t{\b\phi}, \t{\b{\b\phi}})$ as follows
\begin{align}\label{3-periodic-cond-1}
\t\phi=\lambda^{-3}T_{n_3}^3(\t\phi),~~\t{\b{\b\phi}}=\lambda^{-2}T_{n_3}^2(\t\phi),~~\t{\b\phi}=\lambda^{-1}T_{n_3}(\t\phi),
\end{align}
which means the 3-periodic conditions are preserved under the Darboux transformation \eqref{dpmBSQ-DT-1-1-1}.  Furthermore, for the new potentials $(\t v, \t w)$, we have
\begin{subequations}
\begin{align*}
&\t v=\frac{T_{_{n_3}}^2\!(\theta)}{\theta}v\!=\!\mu^3 T_{_{n_3}}\!\!\left(\frac{\b\theta}{T_{_{n_3}}\!(\b\theta)}\right)\!v\!=
\!\mu^3\frac{\b{\b\theta}}{T_{_{n_3}}\!\!\left(\b{\b\theta}\right)\!}v,\\
&\t w=\frac{T_{_{n_3}}\!(\theta)}{\theta}w\!=\!\mu^3\frac{\b\theta}{T_{_{n_3}}^2\!(\b\theta)}w\!=\!
T_{_{n_3}}\!\left(\frac{T_{_{n_3}}\!\left(\b{\b\theta}\right)}{\b{\b\theta}}\right)\!w,
\end{align*}
\end{subequations}
which means each component itself of the vector $(\theta, \b\theta, \b{\b\theta})$ can be used to express the new potentials $(\t v, \t w)$.

\noindent \emph{\textbf{Remark.}} We may also describe the gauge transformation of $\bm\Phi=(\phi, \b\phi, \b{\b\phi})^T$ in \eqref{dpmBSQ-DT-1-1-1} in matrix form as follows
\begin{equation*}
\mathrm {DT}^{\bm\Theta}:~~
\bm\Phi\rightarrow\left(
\begin{array}{ccc}
-\mu\b\theta\theta^{-1} & \lambda & 0\\
0 &  -\mu\b{\b\theta}~\b\theta^{-1} & \lambda\\
\lambda & 0 & -\mu\theta\b{\b\theta}^{-1}
\end{array}\right)\bm\Phi.
\end{equation*}
For convenience of the construction of binary Darboux transformations, we here write in scalar form shown in \eqref{dpmBSQ-DT-1-1-1}.

As usual, we may obtain closed form expressions for the result of $N$ applications of the above Darboux transformation. Then we have the following.
\begin{prop}\label{prop1N-dpmBSQ}
Let $(\theta^{1}, \b\theta^{1}, \b{\b\theta}^{1})^T, (\theta^{2},\b\theta^{2}, \b{\b\theta}^{2})^T,  \dots,  (\theta^{N}, \b\theta^{N}, \b{\b\theta}^{N})^T$, having 3-periodic property  $\theta^{k}\!\!=\!\!\lambda_k\!^{-3}T_{_{n_3}}\!\!\!\!^3(\theta^{k})$, $\b{\b\theta}^{k}\!\!\!=\!\!\lambda_k\!^{-2}T_{_{n_3}}\!\!\!\!^2(\theta^{k})$, $\b\theta^{k}\!\!=\!\!\lambda_k\!^{-1}T_{_{n_3}}(\theta^{k})$, be $N$ non-zero independent vector solutions of the linear system \eqref{matrix-LP} by taking $\lambda=\lambda_k$, $k=1, 2, \dots, N$, for some $(v, w)$, such that $C_{_{[3]}}(\theta^{1},\theta^{2}, \dots, \theta^{N})\neq0$, $C_{_{[3]}}(\b\theta^{1},\b\theta^{2}, \dots, \b\theta^{N})\neq0$, and $C_{_{[3]}}(\b{\b\theta}^{1},\b{\b\theta}^{2}, \dots, \b{\b\theta}^{N})\neq0$.  Then
\begin{subequations}\label{N-DT1}
\begin{align}
&\phi\rightarrow\t\phi=\frac{C_{_{[3]}}(\theta^{1},\theta^{2}, \dots, \theta^{N},\phi)}{C_{_{[3]}}(\theta^{1},\theta^{2}, \dots, \theta^{N})},
~~\b\phi\rightarrow\t{\b\phi}=\frac{C_{_{[3]}}(\b\theta^{1},\b\theta^{2}, \dots, \b\theta^{N},\b\phi)}{C_{_{[3]}}(\b\theta^{1},\b\theta^{2}, \dots, \b\theta^{N})},
~~\b{\b\phi}\rightarrow\t{\b{\b\phi}}=\frac{C_{_{[3]}}(\b{\b\theta}^{1},\b{\b\theta}^{2}, \dots, \b{\b\theta}^{N},\b{\b\phi})}{C_{_{[3]}}(\b{\b\theta}^{1},\b{\b\theta}^{2}, \dots, \b{\b\theta}^{N})},\label{N-DT1-1}\\
&v\rightarrow \t v\!\!=\!\!\frac{T_{_{n_3}}^2\!\!\left(C_{_{[3]}}(\theta^{1},\theta^{2}, \dots, \theta^{N})\right)}{C_{_{[3]}}(\theta^{1},\theta^{2}, \dots, \theta^{N})}v,~~~~
w\rightarrow \t w\!\!=\!\!\frac{T_{_{n_3}}\!\!\left(C_{_{[3]}}(\theta^{1},\theta^{2}, \dots, \theta^{N})\right)}{C_{_{[3]}}(\theta^{1},\theta^{2}, \dots, \theta^{N})}w,
\end{align}
\end{subequations}
leaves \eqref{matrix-LP} invariant.
\end{prop}

\begin{proof}
Let
\begin{align*}
F:=&C_{_{[3]}}\left(\theta^1, \theta^2, \dots, \theta^N\right)\!=\!|\bm\theta(0),\bm\theta(1),  \dots, \bm\theta(N-1)|,\\
G:=&C_{_{[3]}}\left(\theta^1, \theta^2, \dots, \theta^N,\phi\right)\!=\!|\bm\theta^{^\dag}\!(0),\bm\theta^{^\dag}\!(1), \dots, \bm\theta^{^\dag}\!(N)|,\\
\b F:=&C_{_{[3]}}\left(\b\theta^1, \b\theta^2, \dots, \b\theta^N\right)\!=\!|\b{\bm\theta}(0), \b{\bm\theta}(1),  \dots, \b{\bm\theta}(N-1)|,\\
\b G:=&C_{_{[3]}}\left(\b\theta^1,\b\theta^2, \dots, \b\theta^N,\b\phi\right)\!=\!|\b{\bm\theta}^{\dag}\!(0),\b{\bm\theta}^{\dag}\!(1), \dots,\b{\bm\theta}^{\dag}\!(N)|,\\
\b{\b F}:=&C_{_{[3]}}\left(\b{\b\theta}^1, \b{\b\theta}^2, \dots, \b{\b\theta}^N\right)\!=\!|\b{\bm{\b\theta}}(0), \b{\bm{\b\theta}}(1),  \dots, \b{\bm{\b\theta}}(N-1)|,\\
\b{\b G}:=&C_{_{[3]}}\left(\b{\b\theta}^1,\b{\b\theta}^2, \dots, \b{\b\theta}^N,\b{\b\phi}\right)\!=\!|\b{\bm{\b\theta}}^{\dag}\!(0),\b{\bm{\b\theta}}^{\dag}\!(1), \dots,\b{\bm{\b\theta}}^{\dag}\!(N)|,
\end{align*}
where $\bm\theta\!=\!\left(\theta^1, \theta^2, \dots, \theta^N\right)^{\!T}$, $\bm\theta^{^\dag}\!=\!\left(\theta^1, \theta^2, \dots, \theta^N\!,\phi\right)^{\!T}$, $\b{\bm\theta}\!=\!\left(\b\theta^1, \b\theta^2, \dots, \b\theta^N\right)^{\!T}$, $\b{\bm\theta}^{^\dag}=\left(\b\theta^1, \b\theta^2, \dots, \b\theta^N\!,\b\phi\right)^{\!T}$, $\b{\bm{\b\theta}}\!=\!\left(\b{\b\theta}^1, \b{\b\theta}^2, \dots, \b{\b\theta}^N\right)^{\!T}$, $\b{\bm{\b\theta}}^{^\dag}=\left(\b{\b\theta}^1, \b{\b\theta}^2, \dots, \b{\b\theta}^N\!,\b{\b\phi}\right)^{\!T}$.
Moreover, by the 3-periodic properties
\begin{align*}
&\b\phi=\lambda^{-1}T_{n_3}(\phi),~~\b{\b\phi}=\lambda^{-1}T_{n_3}(\b\phi),~~\phi=\lambda^{-1}T_{n_3}(\b{\b\phi}),\\
&\b\theta^i=\lambda_i^{-1}T_{n_3}(\theta^i),~~\b{\b{\theta}}^i=\lambda_i^{-1}T_{n_3}({\b\theta}^i),~~
\theta^i=\lambda_i^{-1}T_{n_3}({\b{\b\theta}}^i),
\end{align*}
we get the relations
\begin{align*}
\frac{\b G}{\b F}=\lambda^{-1}T_{n_3}(\frac{G}{F}),~~\frac{\b{\b{G}}}{\b{\b{F}}}=\lambda^{-1}T_{n_3}(\frac{\b G}{\b F}),~~\frac{ G}{F}=\lambda^{-1}T_{n_3}(\frac{\b{\b{G}}}{\b{\b{F}}}).
\end{align*}
To verify that \eqref{matrix-LP} is invariant under \eqref{N-DT1} we must show that, for $\kappa=1, 2$,
\begin{subequations}
\begin{align}
G_\kappa T_{_{n_3}}\!(F)=&a_\kappa w_{\kappa}w^{-1}GT_{_{n_3}}\!(F_\kappa)+T_{_{n_3}}\!(G)F_\kappa,  \label{bilinear-1}\\
\b G_\kappa T_{_{n_3}}\!(\b F)=&a_\kappa T_{n_3}\!\!\left(w_{\kappa}w^{-1}\right)\b GT_{_{n_3}}\!(\b F_\kappa)+T_{_{n_3}}\!(\b G)\b F_\kappa, \label{bilinear-2} \\
\b{\b G}_\kappa T_{_{n_3}}\!(\b{\b F})=&a_\kappa T_{n_3}^2\!\!\left(w_{\kappa}w^{-1}\right)\b {\b G}T_{_{n_3}}\!(\b {\b F}_\kappa)+T_{_{n_3}}\!(\b {\b G})\b {\b F}_\kappa.  \label{bilinear-3}
\end{align}
\end{subequations}
For the equation \eqref{bilinear-1}, it is equivalent to, under the backward shift operator $T_{_{n_3}}^{-1}T_{_{n_{_\kappa}}}^{-1}$,
\begin{align}\label{bilinear-form-1}
T_{_{n_3}}^{-1}\!(G) F_{_{\b\kappa}}-a_\kappa T_{n_3}^{-1}\!\!\left(w{w_{\b\kappa}}^{\!\!-1}\right) T_{_{n_3}}^{-1}\!(G_{_{\b\kappa}})F-G_{_{\b \kappa}}T_{_{n_3}}^{-1}\!(F)=0.
\end{align}
From \eqref{matrix-LP}, we can deduce the following formulae which are the basic properties we use in the proof.
\begin{align*}
\bm\theta_{\b\kappa}(l)\!=\!&\bm\theta(l-1)+\sum_{i=0}^{l-2}(-\alpha_{\b\kappa})^{i+1}\bm\theta(l-2-i)+(-\alpha_{\b\kappa})^l \bm\theta_{\b\kappa}(0),  \\
{\bm\theta^\dag}_{\!\b\kappa}(l)\!=\!&\bm\theta^\dag(l-1)+\sum_{i=0}^{l-2}(-\alpha_{\b\kappa})^{i+1}{\bm\theta^\dag}(l-2-i)+(-\alpha_{\b\kappa})^l {\bm\theta^\dag}_{\!\b\kappa}(0),
\end{align*}
where $\alpha_{\b\kappa}=a_\kappa T_{n_3}^{-1}\!\!\left(w{w_{\b\kappa}}^{\!\!-1}\right)$ is a scalar, $\kappa=1, 2$.  Then, for \eqref{bilinear-form-1}, it follows that
\begin{subequations}\label{bilinear-cond}
\begin{align}
T_{_{n_3}}^{-1}(G)=&\left|\bm {\theta^\dag}\!(-1), \bm{\theta^\dag}\!(0), \bm{\theta^\dag}\!(1), \cdots \bm{\theta^\dag}\!(N-1)\right|,    \\
F_{_{\b\kappa}}=&\left|\bm {\theta_{\!\b\kappa}^\dag}(0), \bm{\theta}(0), \bm{\theta}(1), \cdots \bm{\theta}(N-2)\right|,       \\
\alpha_{\b\kappa}T_{_{n_3}}^{-1}\!(G_{_{\b\kappa}})=&-\left|\bm {\theta^\dag}_{\b\kappa}\!(0), \bm{\theta^\dag}\!(-1), \bm{\theta^\dag}\!(0), \cdots \bm{\theta^\dag}\!(N-2)\right|,     \\
F=&\left|\bm{\theta}(0), \bm{\theta}(1), \bm{\theta}(2),  \cdots \bm{\theta}(N-1)\right|,      \\
G_{_{\b \kappa}}=&\left|\bm {\theta_{\b\kappa}^\dag}\!(0), \bm{\theta^\dag}\!(0), \bm{\theta^\dag}\!(1),  \cdots \bm{\theta^\dag}\!(N-1)\right|,  \\
T_{_{n_3}}^{-1}\!(F)=&\left|\bm{\theta}(-1), \bm{\theta}(0), \bm{\theta}(1), \cdots \bm{\theta}(N-2)\right|.
\end{align}
\end{subequations}
Substituting the above Casoratian determinants \eqref{bilinear-cond} into the left-hand side of \eqref{bilinear-form-1},  and using the Laplace theorem, we get
\begin{align*}
\!\!\!\!\!\!\!\!\!\!\!\!\!\!\!\!\text{LHS}\!=\!\left |\!\!
\begin{array}{ccccccccc}
\bm {\theta^\dag}_{\b\kappa}(0) & \bm{\theta^\dag}(-1) & \bm{\theta^\dag}(0) & \cdots & \bm{\theta^\dag}(N-2) & \bm0 & \cdots & \bm0 & \bm{\theta^\dag}(N-1)
\\
\bm {\theta}_{\b\kappa}\!(0) & \bm{\theta}(-1) & \bm0  & \cdots &  \bm0  &  \bm{\theta}(0)     & \cdots                      &     \bm{\theta}(N-2)     & \bm{\theta}(N-1)
\end{array}
\!\!\right |\!=\!0.
\end{align*}
So the equation \eqref{bilinear-1} is proved. Moreover,  the equations \eqref{bilinear-2} and \eqref{bilinear-3} can be obtained from the equation \eqref{bilinear-1} by taking the difference operator $T_{n_3}$ and $T_{n_3}^2$ to act on it, respectively. So the equations \eqref{bilinear-2} and \eqref{bilinear-3}  are right as well.
\end{proof}

Through the 3-periodic conditions of $(\theta^{1}, \b\theta^{1}, \b{\b\theta}^{1})^T, (\theta^{2},\b\theta^{2}, \b{\b\theta}^{2})^T,  \dots,  (\theta^{N}, \b\theta^{N}, \b{\b\theta}^{N})^T$ and $(\phi, \b\phi, \b{\b\phi})$, we can obtain the 3-periodic relationships for the new eigenfunctions $(\t\phi, \t{\b\phi}, \t{\b{\b\phi}})$ same as that shown in \eqref{3-periodic-cond-1}. Similarly, for the new potentials $(\t v, \t w)$, we have
\begin{align*}
&\t v\!\!=\!\!\frac{T_{_{n_3}}^2\!\!\left(C_{_{[3]}}(\theta^{1},\theta^{2}, \dots, \theta^{N})\right)}{C_{_{[3]}}(\theta^{1},\theta^{2}, \dots, \theta^{N})}v
\!\!=\!\!\prod_{i=1}^{N}\!\!\lambda_i^3T_{n_3}\!\!\left(\frac{C_{_{[3]}}(\b\theta^{1},\b\theta^{2}, \dots, \b\theta^{N})}{T_{_{n_3}}\!\!\left(C_{_{[3]}}(\b\theta^{1},\b\theta^{2}, \dots, \b\theta^{N})\right)}\right)v
\!\!=\!\!\prod_{i=1}^{N}\!\!\lambda_i^3\frac{C_{_{[3]}}(\b{\b\theta}^{1},\b{\b\theta}^{2}, \dots, \b{\b\theta}^{N})}{T_{_{n_3}}\!\!\left(C_{_{[3]}}(\b{\b\theta}^{1},\b{\b\theta}^{2}, \dots, \b{\b\theta}^{N})\right)}v,\\
&\t w\!\!=\!\!\frac{T_{_{n_3}}\!\!\left(C_{_{[3]}}(\theta^{1},\theta^{2}, \dots, \theta^{N})\right)}{C_{_{[3]}}(\theta^{1},\theta^{2}, \dots, \theta^{N})}w
\!\!=\!\!\prod_{i=1}^{N}\!\!\lambda_i^3\frac{C_{_{[3]}}(\b\theta^{1},\b\theta^{2}, \dots, \b\theta^{N})}{T_{_{n_3}}^2\!\!\left(C_{_{[3]}}(\b\theta^{1},\b\theta^{2}, \dots, \b\theta^{N})\right)}w
\!\!=\!\!T_{n_3}\!\!\left(\!\!\frac{T_{_{n_3}}\!\!\left(\!\!C_{_{[3]}}(\b{\b\theta}^{1},\b{\b\theta}^{2}, \dots, \b{\b\theta}^{N})\right)}{C_{_{[3]}}(\b{\b\theta}^{1},\b{\b\theta}^{2}, \dots, \b{\b\theta}^{N})}\!\!\right)\!\!w.
\end{align*}

As described in \cite{SNZ-2014}, the Hirota-Miwa equation \eqref{H-M-1} is invariant with respect to the reversal of all lattice directions $n_i\rightarrow -n_i$, i=1,2. But its linear system \eqref{H-M-LP-1} does not have such invariance and so the reflections $n_i\rightarrow -n_i$, i=1,2, acting on \eqref{H-M-LP-1} give a second linear system on the vector eigenfunction $\bm\Psi=(\psi,\b \psi)^T$ (see \cite{SNZ-2014}) .  So here we introduce eigenfunctions  $\psi(n_1,n_2,n_3)$, $\b \psi(n_1, n_2, n_3)$ and $\b{\b{\psi}}(n_1, n_2, n_3)$ and give a 3-periodic condition, as follows
\begin{equation}\label{psi-condition}
\psi=\lambda^{-3}T_{n_3}^{-3}(\psi),\quad \b{\b\psi}=\lambda^{-2} T_{n_3}^{-2}(\psi), \quad \b{\psi}=\lambda^{-1} T_{n_3}^{-1}(\psi),
\end{equation}
where the parameter $\lambda$ serves as the spectral parameter. From \eqref{psi-condition}, we have
\begin{subequations}\label{bpsi-condition}
\begin{align}
\b\psi=\lambda^{-3}T_{n_3}^{-3}(\b\psi),\quad \psi=\lambda^{-2} T_{n_3}^{-2}(\b\psi), \quad \b{\b\psi}=\lambda^{-1} T_{n_3}^{-1}(\b\psi),\\
\b{\b\psi}=\lambda^{-3}T_{n_3}^{-3}(\b{\b\psi}),\quad \b\psi=\lambda^{-2} T_{n_3}^{-2}(\b{\b\psi}), \quad \psi=\lambda^{-1} T_{n_3}^{-1}(\b{\b\psi}).
\end{align}
\end{subequations}
We impose the above 3-periodic conditions of $\bm\Psi(n_1,n_2,n_3)$ on the second linear system of the Hirota-Miwa equation and obtain a second linear system of the discrete modified Boussinesq equation \eqref{d-p-MBSQ} as follows
\begin{subequations}\label{matrix-LP-2}
\begin{gather}
\bm\Psi_{\b1}=\bm U\bm\Psi,\label{matrix-LP-2-1}\\
\bm\Psi_{\b2}=\bm V\bm\Psi, \label{matrix-LP-2-2}
\end{gather}
\end{subequations}
where
\begin{equation}
\bm U=\left(
    \begin{array}{ccc}
      a_1v_{\b1}v^{-1} & \lambda & 0 \\
      0 & a_1w_{\b 1}w^{-1}vv_{\b 1}^{-1} & \lambda  \\
      \lambda & 0 & a_1ww_{\b1} ^{-1}
    \end{array}
  \right)
,~~
\bm V=\left(
    \begin{array}{ccc}
      a_2v_{\b2}v^{-1} & \lambda & 0 \\
      0 & a_2w_{\b 2}w^{-1}vv_{\b 2}^{-1} & \lambda  \\
      \lambda & 0 & a_2ww_{\b2} ^{-1}
    \end{array}
  \right)
.\nonumber
\end{equation}
One then finds the compatible condition
\begin{align}\label{compatibility-matrix-LP-2}
0&=\bm\Psi_{\b1\b2}-\bm\Psi_{\b2\b1}=(\bm U_{\b2}\bm V-\bm V_{\b1}\bm U)\bm\Psi\nonumber\\
&=\!\lambda\!\!\left(
    \begin{array}{ccc}
      0&\frac{(a_1\!v_{\b1}-a_2\!v_{\b2})\!v_{\b1\b2}w-(a_1\!v_{\b2}w_{\b1}-a_2\!v_{\b1}\!w_{\b2})\!v}{v_{\b1}v_{\b2}w}&0\\
      0& 0&\frac{(a_1\!v_{\b2}\!w_{\b1}-a_2\!v_{\b1}\!w_{\b2})w_{\b1\b2}-(a_1\!w_{\b2}-a_2\!w_{\b1})v_{\b1\b2}w}{v_{\b1\b2}w_{\b1}w_{\b2}}\\
      \frac{(a_1\!w_{\b2}-a_2\!w_{\b1})v-(a_1\!v_{\b1}-a_2\!v_{\b2})w_{\b1\b2}}{vw_{\b1\b2}}&0&0 \\
    \end{array}
  \right)\!\!\bm\Psi.
\end{align}
By applying the forward shift operators $T_{n_1}T_{n_2}$ on the entries $(1, 2)$, $(2, 3)$ and $(3, 1)$ of the matrix in \eqref{compatibility-matrix-LP-2}, we get the discrete modified Boussinesq equation \eqref{d-p-MBSQ}.

\noindent\emph{\textbf{Remark.}} This linear system \eqref{matrix-LP-2} is not the adjoint form of the linear system \eqref{matrix-LP} under the conception of the discrete adjoint operator introduced in \cite{N-1997}. Therefore, unlike in the continuous case, the binary Darboux transformations constructed in the current paper involve two different linear systems without adjoint relationship.

The Darboux transformation for the second linear system \eqref{matrix-LP-2} is as follows.
\begin{prop}\label{prop2-dpmBSQ}
Let $\bm\rho=(\rho,\b\rho, \b{\b\rho})^T$ having the 3-periodic property $\rho=\mu^{-3}T_{n_3}^{-3}(\rho),\quad \b{\b\rho}=\mu^{-2} T_{n_3}^{-2}(\rho), \quad \b{\rho}=\mu^{-1} T_{n_3}^{-1}(\rho)$, be a non-zero vector solution of the linear system \eqref{matrix-LP-2} by taking $\lambda=\mu$ for some $(v, w)$, then
\begin{subequations}\label{dpmBSQ-DT-2-1}
\begin{align}
\mathrm {DT}^{\bm\rho}:
&~~\psi\rightarrow\t\psi=\frac{C_{_{[\b 3]}}(\rho, \psi)} {\rho}, ~~~~\b\psi\rightarrow\t{\b\psi}=\frac{C_{_{[\b 3]}}(\b\rho, \b\psi)} {\b\rho} , ~~~~\b{\b\psi}\rightarrow\t{\b{\b\psi}}=\frac{C_{_{[\b 3]}}(\b{\b\rho}, \b{\b\psi})} {\b{\b\rho}} ,\label{dpmBSQ-DT-2-1-1}\\
&~~v\rightarrow \t v=\frac{T_{_{n_3}}^{-1}(\rho)}{\rho}v,
~~w\rightarrow \t w=\frac{T_{_{n_3}}^{-2}(\rho)}{\rho}w
\end{align}
\end{subequations}
leaves \eqref{matrix-LP-2} invariant.
\end{prop}
The proof of \textbf{Prop.\ref{prop2-dpmBSQ}} is a straightforward computation and similar to that of \textbf{Prop.\ref{prop1N-dpmBSQ}} when $N=1$.  So we do not give details here.

Through the 3-periodic conditions of $(\rho, \b\rho, \b{\b\rho})$ and $(\psi, \b\psi, \b{\b\psi})$, we can obtain the 3-periodic relationships for the new eigenfunctions $(\t\psi, \t{\b\psi}, \t{\b{\b\psi}})$ as follows
\begin{align}\label{3-periodic-cond-3}
\t\psi=\lambda^{-3}T_{n_3}^{-3}(\t\psi),~~\t{\b{\b\psi}}=\lambda^{-2}T_{n_3}^{-2}(\t\psi),
~~\t{\b\psi}=\lambda^{-1}T_{n_3}^{-1}(\t\psi),
\end{align}
which means the 3-periodic conditions are preserved under the Darboux transformation \eqref{dpmBSQ-DT-2-1-1}.  Furthermore, for the new potentials $(\t v, \t w)$, we have
\begin{subequations}
\begin{align*}
&\t v=\frac{T_{_{n_3}}^{-1}(\rho)}{\rho}v=\mu^3\frac{\b\rho}{T_{_{n_3}}^{-2}(\b\rho)}v
=T_{_{n_3}}^{-1}\!\!\left(\frac{T_{_{n_3}}^{-1}(\b{\b\rho})}{\b{\b\rho}}\right)\!\!v,\\
&\t w=\frac{T_{_{n_3}}^{-2}(\rho)}{\rho}w=\mu^3T_{_{n_3}}^{-1}\!\!\left(\frac{\b\rho}{T_{_{n_3}}^{-1}(\b\rho)}\right)\!\!w
=\mu^{3}\frac{\b{\b\rho}}{T_{_{n_3}}^{-1}(\b{\b\rho})}w,
\end{align*}
\end{subequations}
which means each component itself of the vector $(\rho, \b\rho, \b{\b\rho})$ can be used to express the new potentials $(\t v, \t w)$.

Next we write down the closed form expression for the result of $N$ applications of the above Darboux transformation, which give solutions in Casoratian determinant form.

\begin{prop}\label{prop2N-dpmBSQ}
Let $(\rho^{1}, \b\rho^{1},\b{\b{\rho}}^{1})^T, (\rho^{2},\b\rho^{2},\b{\b{\rho}}^{2})^T,  \dots,  (\rho^{N}, \b\rho^{N},\b{\b{\rho}}^{N})^T$, having the 3-periodic property  $\rho^{k}\!=\!\lambda_k^{-3}T_{_{n_3}}^{-3}(\rho^{k}), \b{\b\rho}^{k}\!=\!\lambda_k^{-2}T_{_{n_3}}^{-2}(\rho^{k}),\b \rho^{k}\!=\!\lambda_k^{-1}T_{_{n_3}}^{-1}(\rho^{k})$, be $N$ non-zero independent vector solutions of the linear system \eqref{matrix-LP-2} by taking $\lambda=\lambda_k$, $k=1, 2, \dots, N$, for some $(v, w)$, such that $C_{_{[\b 3]}}(\rho^{1},\rho^{2}, \dots, \rho^{N})\neq0$, $C_{_{[\b 3]}}(\b\rho^{1},\b\rho^{2}, \dots, \b\rho^{N})\neq0$, $C_{_{[\b 3]}}(\b{\b\rho}^{1},\b{\b\rho}^{2}, \dots, \b{\b\rho}^{N})\neq0$.  Then
\begin{subequations}\label{N-DT2}
\begin{align}
&\psi\rightarrow\t\psi=\frac{C_{_{[\b 3]}}(\rho^{1},\rho^{2}, \dots, \rho^{N},\psi)}{C_{_{[\b 3]}}(\rho^{1},\rho^{2}, \dots, \rho^{N})},
~~\b\psi\rightarrow\t{\b\psi}=\frac{C_{_{[\b 3]}}(\b\rho^{1},\b\rho^{2}, \dots, \b\rho^{N},\b\psi)}{C_{_{[\b 3]}}(\b\rho^{1},\b\rho^{2}, \dots, \b\rho^{N})},
~~\b{\b\psi}\rightarrow\t{\b{\b\psi}}=\frac{C_{_{[\b 3]}}(\b{\b\rho}^{1},\b{\b\rho}^{2}, \dots, \b{\b\rho}^{N},\b{\b\psi})}{C_{_{[\b 3]}}(\b{\b\rho}^{1},\b{\b\rho}^{2}, \dots, \b{\b\rho}^{N})},\label{N-DT2-1}
\\
&v\rightarrow \t v=\frac{T_{_{n_3}}^{-1}\!\!\left(C_{_{[\b 3]}}(\rho^{1},\rho^{2}, \dots, \rho^{N})\right)}{C_{_{[\b 3]}}(\rho^{1},\rho^{2}, \dots, \rho^{N})}v,~~w\rightarrow \t w=\frac{T_{_{n_3}}^{-2}\!\!\left(C_{_{[\b 3]}}(\rho^{1},\rho^{2}, \dots, \rho^{N})\right)}{C_{_{[\b 3]}}(\rho^{1},\rho^{2}, \dots, \rho^{N})}w,
\end{align}
\end{subequations}
leaves \eqref{matrix-LP-2} invariant.
\end{prop}
The proof of \textbf{Prop.\ref{prop2N-dpmBSQ}} is similar to the proof of \textbf{Prop.\ref{prop1N-dpmBSQ}}, so it is omitted here.

Through the 3-periodic conditions of $(\rho^{1}, \b\rho^{1}, \b{\b\rho}^{1})^T, (\rho^{2},\b\rho^{2}, \b{\b\rho}^{2})^T,  \dots,  (\rho^{N}, \b\rho^{N}, \b{\b\rho}^{N})^T$ and $(\psi, \b\psi, \b{\b\psi})$, we can obtain the 3-periodic relationships for the new eigenfunctions $(\t\psi, \t{\b\psi}, \t{\b{\b\psi}})$ same as that shown in \eqref{3-periodic-cond-3}. Similarly, for the new potentials $(\t v, \t w)$, we have
\begin{align*}
&\t v=\frac{T_{_{n_3}}^{-1}\!\!\left(C_{_{[\b 3]}}(\rho^{1},\rho^{2}, \dots, \rho^{N})\right)}{C_{_{[\b 3]}}(\rho^{1},\rho^{2}, \dots, \rho^{N})}v
=\!\!\prod_{k=1}^{N}\!\!\lambda_k^3\frac{C_{_{[\b 3]}}(\b\rho^{1},\b\rho^{2}, \dots, \b\rho^{N})}{T_{_{n_3}}^{-2}\!\!\left(C_{_{[\b 3]}}(\b\rho^{1}, \b\rho^{2}, \dots, \b\rho^{N})\right)}v
=T_{_{n_3}}^{-1}\!\!\left(\!\!\frac{T_{_{n_3}}^{-1}\!\!\left(C_{_{[\b 3]}}(\b{\b\rho}^{1}, \b{\b\rho}^{2}, \dots, \b{\b\rho}^{N})\right)}{C_{_{[\b 3]}}(\b{\b\rho}^{1},\b{\b\rho}^{2}, \dots, \b{\b\rho}^{N})}\!\!\right)\!v,\\
&\t w=\frac{T_{_{n_3}}^{-2}\!\!\left(C_{_{[\b 3]}}(\rho^{1},\rho^{2}, \dots, \rho^{N})\right)}{C_{_{[\b 3]}}(\rho^{1},\rho^{2}, \dots, \rho^{N})}w
=\!\!\prod_{k=1}^{N}\!\!\lambda_k^3T_{_{n_3}}^{-1}\!\!\left(\!\!\frac{C_{_{[\b 3]}}(\b\rho^{1},\b\rho^{2}, \dots, \b\rho^{N})}{T_{_{n_3}}^{-1}\!\!\left(C_{_{[\b 3]}}(\b\rho^{1}, \b\rho^{2}, \dots, \b\rho^{N})\right)}\!\!\right)\!\!w
=\!\!\prod_{k=1}^{N} \lambda_k^3\frac{C_{_{[\b 3]}}(\b{\b\rho}^{1},\b{\b\rho}^{2}, \dots, \b{\b\rho}^{N})}{T_{_{n_3}}^{-1}\!\!\left(C_{_{[\b 3]}}(\b{\b\rho}^{1}, \b{\b\rho}^{2}, \dots, \b{\b\rho}^{N})\right)}w.
\end{align*}
\section{Binary Darboux transformations}\label{sec-4}

It is possible to construct the binary Darboux transformations depending on two eigenfunctions of two different linear systems. We denote the first linear system \eqref{matrix-LP} by $L^{\textcircled{1}}$ and the second linear system \eqref{matrix-LP-2} by $L^{\textcircled{2}}$.  Let $L^{\textcircled{1}}$ and $\widehat{L}^{\textcircled{1}}$ be the linear systems linked to Darboux transformations $\mathrm {DT}^{\bm\Theta}$, $\mathrm {DT}^{\widehat{\bm\Theta}}$ to the same linear system $\widetilde{L}^{\textcircled{1}}$, where $\widehat{L}^{\textcircled{1}}$ is a copy of the same linear system $L^{\textcircled{1}}$ depending on the potentials $(\hat{v}, \hat{w})$ and $\tilde{L}^{\textcircled{1}}$ depending on the potentials $(\tilde{v}, \tilde{w})$ is not necessarily of the same type as $L^{\textcircled{1}}$.   Schematically one can construct a binary Darboux transformation $\mathrm{BDT}$ as a transformation made up of two successive Darboux transformations in the following special manner:
\begin{figure}[H]
\centering
\includegraphics[width=2.0in,height=1.0in]{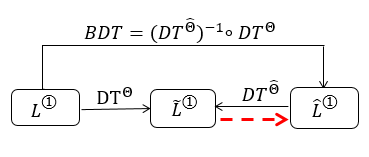}
\caption{Construction of the binary Darboux transformation}
\label{bDT}
\end{figure}
\noindent In order to define $\mathrm {DT}^{\widehat{\bm\Theta}}$ one needs an eigenfunction $\widehat{\bm\Theta}$ of $\widehat{L}^{\textcircled{1}}$. This problem can be got around by using the second linear system $L^{\textcircled{2}}$. If $(\phi,\b\phi, \b{\b\phi})^T$ and $(\psi, \b\psi, \b{\b\psi})^T$ are non-zero vector eigenfunctions of the linear system $L^{\textcircled{1}}$ and $L^{\textcircled{2}}$, respectively, corresponding to spectrum parameter $\lambda$, then we define the eigenfunction potentials $\omega(\phi, \psi)$, $\omega(\b\phi, \b\psi)$, $\omega(\b{\b\phi}, \b{\b\psi})$ as follows
\begin{subequations}\label{lmBSQ-omega}
\begin{align}
\Delta_3(\omega(\phi, \psi))&=\phi T_{_{n_3}}\!(\psi),  \label{lmBSQ-omega-1}\\
\Delta_3(\omega(\b\phi, \b\psi))&=\b\phi T_{_{n_3}}\!(\b{\b\psi}), \label{lmBSQ-omega-2}\\
\Delta_3(\omega(\b{\b\phi}, \b{\b\psi}))&=\b{\b\phi }T_{_{n_3}}\!(\b{\psi}).  \label{lmBSQ-omega-3}
\end{align}
\end{subequations}
From \eqref{lmBSQ-omega}, we have if  $(\theta,\b\theta, \b{\b\theta})^T$ and $(\rho,\b\rho, \b{\b\rho})^T$ are non-zero vector eigenfunctions of the linear system $L^{\textcircled{1}}$ and $L^{\textcircled{2}}$, respectively, corresponding to spectrum parameter $\mu$, then the eigenfunction $\widehat{\bm\Theta}=(\widehat{\theta},\widehat{\b{\theta}}, \widehat{\b{\b{\theta}}})^T$ is given by
\begin{align*}
\widehat{\theta}=\theta\omega(\theta, \rho)^{-1},\\
\widehat{\b\theta}=\b\theta\omega(\b\theta, \b\rho)^{-1},\\
\widehat{\b{\b\theta}}=\b{\b\theta}\omega(\b{\b\theta}, \b{\b\rho})^{-1}.
\end{align*}
We can then construct the binary Darboux transformation (in difference operator form) explicitly as
\begin{align*}
BDT=\left(1 - \theta\omega(\theta, \rho)^{-1}\bigtriangleup_3^{-1}T_{n_3}(\rho), 1 - \b\theta\omega(\b\theta, \b\rho)^{-1}\bigtriangleup_3^{-1}T_{n_3}(\b\rho),1 - \b{\b\theta}\omega(\b{\b\theta}, \b{\b\rho})^{-1}\bigtriangleup_3^{-1}T_{n_3}(\b{\b\rho})\right)^T.
\end{align*}
Analogous to the above procedure, one can construct the binary Darboux transformation for the second linear system $L^{\textcircled{2}}$, denoted by $\mathrm{aBDT}$, as follows
\begin{align*}
aBDT=\left(1 - \rho\omega(\theta, \rho)^{-1}\bigtriangleup_3^{-1}T_{n_3}(\theta), 1 - \b\rho\omega(\b\theta, \b\rho)^{-1}\bigtriangleup_3^{-1}T_{n_3}(\b\theta),1 - \b{\b\rho}\omega(\b{\b\theta}, \b{\b\rho})^{-1}\bigtriangleup_3^{-1}T_{n_3}(\b{\b\theta})\right)^T.
\end{align*}

By using the 3-periodic properties in \eqref{phi-condition} and \eqref{psi-condition}, we have the reduction conditions for $(\omega,\b\omega, \b{\b\omega})^T$ as follows
\begin{subequations}
\begin{align*}
\omega(\phi, \psi)=T_{_{n_3}}^3(\omega(\phi, \psi)),~~\omega(\b\phi, \b\psi))=T_{_{n_3}}(\omega(\phi, \psi)),~~\omega(\b{\b\phi}, \b{\b\psi}))=T_{_{n_3}}^2(\omega(\phi, \psi)),\\
\omega(\b\phi, \b\psi)=T_{_{n_3}}^3(\omega(\b\phi, \b\psi)),~~\omega(\b{\b\phi}, \b{\b\psi}))=T_{_{n_3}}(\omega(\b\phi, \b\psi)),~~\omega({\phi}, \psi))=T_{_{n_3}}^2(\omega(\b\phi, \b\psi)),\\
\omega(\b{\b\phi}, \b{\b\psi})=T_{_{n_3}}^3(\omega(\b{\b\phi}, \b{\b\psi})),~~\omega(\phi, \psi))=T_{_{n_3}}(\omega(\b{\b\phi}, \b{\b\psi})),~~\omega(\b{\phi}, \b{\psi}))=T_{_{n_3}}^2(\omega(\b{\b\phi}, \b{\b\psi})).
\end{align*}
\end{subequations}
Furthermore,
\begin{subequations}
\begin{align*}
\omega(\phi, \rho)=\left(\frac{\mu}{\lambda}\right)^3\!\!T_{_{n_3}}^3(\omega(\phi, \rho)),~~\omega(\b\phi, \b\rho)=\left(\frac{\mu}{\lambda}\right)T_{_{n_3}}(\omega(\phi, \rho)),~~\omega(\b{\b\phi}, \b{\b\rho}))=\left(\frac{\mu}{\lambda}\right)^2T_{_{n_3}}^2(\omega(\phi, \rho)),\\
\omega(\b\phi, \b\rho)=\left(\frac{\mu}{\lambda}\right)^3\!\!T_{_{n_3}}^3(\omega(\b\phi, \b\rho)),~~\omega(\b{\b\phi}, \b{\b\rho})=\left(\frac{\mu}{\lambda}\right)T_{_{n_3}}(\omega(\b\phi, \b\rho)),~~\omega({\phi}, {\rho}))=\left(\frac{\mu}{\lambda}\right)^2T_{_{n_3}}^2(\omega(\b\phi, \b\rho)),\\
\omega(\b{\b\phi}, \b{\b\rho})=\left(\frac{\mu}{\lambda}\right)^3\!\!T_{_{n_3}}^3(\omega(\b{\b\phi}, \b{\b\rho})),~~\omega(\phi, \rho)=\left(\frac{\mu}{\lambda}\right)T_{_{n_3}}(\omega(\b{\b\phi}, \b{\b\rho})),~~\omega(\b{\phi}, \b{\rho}))=\left(\frac{\mu}{\lambda}\right)^2T_{_{n_3}}^2(\omega(\b{\b\phi}, \b{\b\rho})),
\end{align*}
\end{subequations}
\begin{subequations}
\begin{align*}
\omega(\theta, \psi)=\left(\frac{\lambda}{\mu}\right)^3\!\!T_{_{n_3}}^3(\omega(\theta, \psi)),~~\omega(\b\theta, \b\psi)=\left(\frac{\lambda}{\mu}\right)\!T_{_{n_3}}(\omega(\theta, \psi)),~~\omega(\b{\b\theta}, \b{\b\psi}))=\left(\frac{\lambda}{\mu}\right)^2\!\!T_{_{n_3}}^2(\omega(\theta, \psi)),\\
\omega(\b\theta, \b\psi)=\left(\frac{\lambda}{\mu}\right)^3\!\!T_{_{n_3}}^3(\omega(\b\theta, \b\psi)),~~\omega(\b{\b\theta}, \b{\b\psi})=\left(\frac{\lambda}{\mu}\right)\!T_{_{n_3}}(\omega(\b\theta, \b\psi)),~~\omega({\theta}, {\psi}))=\left(\frac{\lambda}{\mu}\right)^2\!\!T_{_{n_3}}^2(\omega(\b\theta,\b\psi)),\\
\omega(\b{\b\theta}, \b{\b\psi})=\left(\frac{\lambda}{\mu}\right)^3\!\!T_{_{n_3}}^3(\omega(\b{\b\theta}, \b{\b\psi})),~~\omega(\theta, \psi)=\left(\frac{\lambda}{\mu}\right)\!T_{_{n_3}}(\omega(\b{\b\theta}, \b{\b\psi})),~~\omega(\b{\theta}, \b{\psi}))=\left(\frac{\lambda}{\mu}\right)^2\!\!T_{_{n_3}}^2(\omega(\b{\b\theta}, \b{\b\psi})),
\end{align*}
\end{subequations}
and
\begin{subequations}
\begin{align*}
\omega(\theta^k, \rho^l)=\left(\frac{\lambda_l}{\lambda_k}\right)^3\!\!T_{_{n_3}}^3(\omega(\theta^k, \rho^l)),~~\omega(\b\theta^k, \b\rho^l))=\left(\frac{\lambda_l}{\lambda_k}\right)\!\!T_{_{n_3}}(\omega(\theta^k, \rho^l)),~~\omega(\b{\b\theta}^k, \b{\b\rho}^l))=\left(\frac{\lambda_l}{\lambda_k}\right)^2\!\!T_{_{n_3}}^2(\omega(\theta^k, \rho^l)),\\
\omega(\b\theta^k, \b\rho^l)=\left(\frac{\lambda_l}{\lambda_k}\right)^3\!\!T_{_{n_3}}^3(\omega(\b\theta^k, \b\rho^l)),~~\omega(\b{\b\theta}^k, \b{\b\rho}^l))=\left(\frac{\lambda_l}{\lambda_k}\right)\!\!T_{_{n_3}}(\omega(\b\theta^k, \b\rho^l)),~~\omega({\theta^k}, \rho^l))=\left(\frac{\lambda_l}{\lambda_k}\right)^2\!\!T_{_{n_3}}^2(\omega(\b\theta^k, \b\rho^l)),\\
\omega(\b{\b\theta}^k, \b{\b\rho}^l)=\left(\frac{\lambda_l}{\lambda_k}\right)^3\!\!T_{_{n_3}}^3(\omega(\b{\b\theta}^k, \b{\b\rho}^l)),~~\omega(\theta^k, \rho^l))=\left(\frac{\lambda_l}{\lambda_k}\right)\!\!T_{_{n_3}}(\omega(\b{\b\theta}^k, \b{\b\rho}^l)),~~\omega(\b{\theta}^k, \b{\rho}^l))=\left(\frac{\lambda_l}{\lambda_k}\right)^2\!\!T_{_{n_3}}^2(\omega(\b{\b\theta}^k, \b{\b\rho}^l)).
\end{align*}
\end{subequations}
The above expressions are useful when we check the binary Darboux transformation.

The following proposition gives the binary Darboux transformation of the discrete modified Boussinesq equation.

\begin{prop}\label{binary-lmBSQ}
For some $(v, w)$, let $(\theta,\b\theta, \b{\b\theta})^T$ and $(\phi,\b\phi, \b{\b\phi})^T$ be two non-zero vector solutions of the linear system \eqref{matrix-LP} , respectively, corresponding to spectrum parameters $\mu$ and $\lambda$; $(\rho,\b\rho, \b{\b\rho})^T$ and $(\psi, \b\psi, \b{\b\psi})^T$ be two non-zero vector solutions of the linear system \eqref{matrix-LP-2}, respectively, corresponding to spectrum parameters $\mu$ and $\lambda$, then
\begin{subequations}\label{lmBSQ-bDT1}
\begin{align}
\mathrm{BDT}:&~\phi\rightarrow\widehat\phi=\phi-\theta \omega(\theta, \rho)^{-1}\omega(\phi,\rho),~~\b\phi\rightarrow\widehat{\b\phi}=\b\phi-\b\theta \omega(\b\theta, \b\rho)^{-1}\omega(\b\phi,\b\rho),~~\b{\b\phi}\rightarrow\widehat{\b{\b\phi}}=\b{\b\phi}-\b{\b\theta} \omega(\b{\b\theta}, \b{\b\rho})^{-1}\omega(\b{\b\phi},\b{\b\rho}),\label{lmBSQ-bDT-1}\\
&~~v\rightarrow\widehat{v}=\frac{T_{_{n_3}}^2\!(\omega(\theta, \rho))}{\omega(\theta, \rho)}v,\quad w\rightarrow\widehat{w}=\frac{T_{_{n_3}}\!(\omega(\theta, \rho))}{\omega(\theta, \rho)}w
\end{align}
\end{subequations}
and
\begin{subequations}\label{lmBSQ-bDT2}
\begin{align}
\mathrm{aBDT}:&~\psi\rightarrow\widehat{\psi}\psi-\rho \omega(\theta, \rho)^{-1}\omega(\theta,\psi),~~\b\psi\rightarrow\widehat{\b\psi}=\b\psi-\b\rho\omega(\b\theta, \b\rho)^{-1}\omega(\b\theta,\b\psi),~~\b{\b\psi}\rightarrow\widehat{\b{\b\psi}}=\b{\b\psi}-\b{\b\rho}\omega(\b{\b\theta}, \b{\b\rho})^{-1}\omega(\b{\b\theta},\b{\b\psi}),\label{lmBSQ-bDT-2}\\
&~~v\rightarrow\widehat{v}=\frac{T_{_{n_3}}^2\!(\omega(\theta, \rho))}{\omega(\theta, \rho)}v,\quad w\rightarrow\widehat{w}=\frac{T_{_{n_3}}\!(\omega(\theta, \rho))}{\omega(\theta, \rho)}w,
\end{align}
\end{subequations}
leave \eqref{matrix-LP} and \eqref{matrix-LP-2} invariant,  respectively.
\end{prop}

The proof of \textbf{Prop.\ref{binary-lmBSQ}} is same with that of \textbf{Prop.\ref{N-binary-lmBSQ}} when $N=1$.  So we do not give details here.

Through the 3-periodic properties of $(\theta, \b\theta, \b{\b\theta})$, $(\phi, \b\phi, \b{\b\phi})$, $(\rho, \b\rho, \b{\b\rho})$ and $(\psi, \b\psi, \b{\b\psi})$, we can obtain the 3-periodic relationships for the new eigenfunctions $(\widehat\phi, \widehat{\b\phi}, \widehat{\b{\b\phi}})$ and $(\widehat\psi, \widehat{\b\psi}, \widehat{\b{\b\psi}})$ as follows
\begin{subequations}\label{3-periodic-cond-4}
\begin{align}
\h\phi\!=\!\lambda^{-2}T_{_{n_3}}^2\!\!\left(\h{\b\phi}\right)\!=\!\lambda^{-1}T_{_{n_3}}\!\!\left(\h{\b{\b\phi}}\right),\quad
&\h\psi\!=\!\lambda^{-2}T_{_{n_3}}^{-2}\!\!\left(\h{\b\psi}\right)\!=\!\lambda^{-1}T_{_{n_3}}^{-1}\!\!\left(\h{\b{\b\psi}}\right), \\
\h{\b\phi}\!=\!\lambda^{-1}T_{_{n_3}}\!\!\left(\h\phi\right)\!=\!\lambda^{-2}T_{_{n_3}}^2\!\!\left(\h{\b{\b\phi}}\right),\quad
&\h{\b\psi}\!=\!\lambda^{-1}T_{_{n_3}}^{-1}\!\!\left(\h\psi\right)\!=\!\lambda^{-2}T_{_{n_3}}^{-2}\!\!\left(\h{\b{\b\psi}}\right),\\
\h{\b{\b\phi}}\!=\!\lambda^{-2}T_{_{n_3}}^2\!\!\left(\h\phi\right)\!=\!\lambda^{-1}T_{_{n_3}}\!\!\left(\h{\b\phi}\right),\quad
&\h{\b{\b\psi}}\!=\!\lambda^{-2}T_{_{n_3}}^{-2}\!\!\left(\h\psi\right)\!=\!\lambda^{-1}T_{_{n_3}}^{-1}\!\!\left(\h{\b\psi}\right),
\end{align}
\end{subequations}
which means the 3-periodic properties are preserved under the binary Darboux transformation \eqref{lmBSQ-bDT-1} and \eqref{lmBSQ-bDT-2}.  Furthermore, for the new potentials $(\widehat v, \widehat w)$, we have
\begin{align*}
\widehat v=\frac{T_{_{n_3}}^2\!(\omega(\theta, \rho))}{\omega(\theta, \rho)}v=T_{_{n_3}}\left(\frac{\omega(\b{\theta}, \b{\rho})}{T_{_{n_3}}\!(\omega(\b{\theta}, \b{\rho}))}\right)v=\frac{\omega(\b{\b\theta}, \b{\b\rho})}{T_{_{n_3}}\!(\omega(\b{\b\theta}, \b{\b\rho}))}v,\\
\widehat w=\frac{T_{_{n_3}}\!(\omega(\theta, \rho))}{\omega(\theta, \rho)}w=\frac{\omega(\b\theta, \b\rho)}{T_{_{n_3}}^2\!(\omega(\b\theta, \b\rho))}w=T_{_{n_3}}\left(\frac{T_{_{n_3}}\!(\omega(\b{\b\theta}, \b{\b\rho}))}{\omega(\b{\b\theta}, \b{\b\rho})}\right)w.
\end{align*}

The $N$-fold iteration of these binary Darboux transformations are given below.

\begin{prop}\label{N-binary-lmBSQ}
Let $(\theta^{1}\!, \b\theta^{1}, \b{\b\theta}^1)^T\!, (\theta^{2}\!,\b\theta^{2}, \b{\b\theta}^2)^T\!,  \dots,  (\theta^{N}\!, \b\theta^{N}, \b{\b\theta}^N)^T\!$ and $(\rho^{1}\!, \b\rho^{1}, \b{\b\rho}^1)^T\!, (\rho^{2}\!,\b\rho^{2}, \b{\b\rho}^2)^T\!,  \dots,  (\rho^{N}\!, \b\rho^{N}, \b{\b\rho}^N)^T\!$, having the 3-periodic properties  $\theta^{k}\!\!=\!\!\lambda_k\!^{-3}T_{_{n_3}}\!\!\!\!^3(\theta^{k})$, $\b{\b\theta}^{k}\!\!\!=\!\!\lambda_k\!^{-2}T_{_{n_3}}\!\!\!\!^2(\theta^{k})$, $\b\theta^{k}\!\!=\!\!\lambda_k\!^{-1}T_{_{n_3}}(\theta^{k})$ and $\rho^{k}\!=\!\lambda_k^{-3}T_{_{n_3}}^{-3}(\rho^{k}), \b{\b\rho}^{k}\!=\!\lambda_k^{-2}T_{_{n_3}}^{-2}(\rho^{k}),\b \rho^{k}\!=\!\lambda_k^{-1}T_{_{n_3}}^{-1}(\rho^{k})$,  with $\lambda=\lambda_k$, $k=1, 2, \dots, N$, be $N$ independent non-zero vector eigenfunctions of the linear systems \eqref{matrix-LP} and \eqref{matrix-LP-2} for some $(v, w)$. Then
\begin{subequations}
\begin{align}\label{N-binary-lmBSQ-1}
&\phi\!\rightarrow\!\h\phi\!=\!\begin{vmatrix}
\bm\Omega(\bm\theta, \bm \rho^T\!) &\!\!\! \bm\theta\\
\bm\Omega(\phi, \bm \rho^T) & \!\!\!\phi \\
\end{vmatrix}
\!|\bm\Omega(\bm\theta, \bm \rho^T\!) |^{-1}\!,
~
\b\phi\!\rightarrow\!\h{\b\phi}\!=\!\begin{vmatrix}
\bm\Omega(\bm{\b\theta}, \bm{\b \rho}^T\!) &\!\! \!\bm{\b\theta}\\
\bm\Omega(\b\phi, \bm {\b\rho}^T\!) &\!\! \!\b\phi \\
\end{vmatrix}
\!|\bm\Omega(\bm{\b\theta}\!, \bm {\b\rho^T}) |^{-1}\!,
~
\b{\b\phi}\!\rightarrow\!\h{\b{\b\phi}}\!=\!\begin{vmatrix}
\bm\Omega(\bm{\b{\b\theta}}, \bm{\b{\b \rho}}^T\!) &\!\! \!\bm{\b{\b\theta}}\\
\bm\Omega(\b{\b\phi}, \bm {\b{\b\rho}}^T\!) &\!\! \!\b{\b\phi} \\
\end{vmatrix}
\!|\bm\Omega(\bm{\b{\b\theta}}\!, \bm {\b{\b\rho}^T}) |^{-1}\!,
\\
&v\!\rightarrow\!\h v\!=\!\frac{|T_{_{n_3}}^2\!\!\!\left(\bm\Omega(\bm\theta, \bm \rho^T\!)\right)\!\!|}{|\bm\Omega(\bm\theta, \bm \rho^T\!)\!|}v,\quad
w\!\rightarrow\!\h w\!=\!\frac{|T_{_{n_3}}\!\!\!\left(\bm\Omega(\bm\theta, \bm \rho^T\!)\right)\!\!|}{|\bm\Omega(\bm\theta, \bm \rho^T\!)\!|}w,
\end{align}
\end{subequations}
and
\begin{subequations}
\begin{align}\label{N-binary-lmBSQ-2}
&\psi\!\rightarrow\!\h\psi\!=\!\begin{vmatrix}
\bm\Omega(\bm\theta^T, \bm \rho) &\!\!\! \bm\rho\\
\bm\Omega(\bm\theta^T, \psi) &\! \!\!\psi \\
\end{vmatrix}
\!|\bm\Omega(\bm\theta^T\!\!\!, \bm \rho) |^{-1},
~
\b\psi\!\rightarrow\!\h{\b\psi}\!=\!\begin{vmatrix}
\bm\Omega(\bm{\b\theta}^T, \bm{\b \rho}) &\!\!\! \bm{\b\rho}\\
\bm\Omega(\bm {\b\theta}^T,\b\psi) &\!\!\! \b\psi \\
\end{vmatrix}
\!|\bm\Omega(\bm{\b\theta}^T\!\!\!, \bm {\b\rho}) |^{-1},
~
\b{\b\psi}\!\rightarrow\!\h{\b{\b\psi}}\!=\!\begin{vmatrix}
\bm\Omega(\bm{\b{\b\theta}}^T, \bm{\b {\b\rho}}) &\!\!\! \bm{\b{\b\rho}}\\
\bm\Omega(\bm {\b{\b\theta}}^T,\b{\b\psi}) &\!\!\! \b{\b\psi} \\
\end{vmatrix}
\!|\bm\Omega(\bm{\b{\b\theta}}^T\!\!\!, \bm {\b{\b\rho}}) |^{-1},
\\
&v\!\rightarrow\!\h v\!=\!\frac{|T_{_{n_3}}^2\!\!\!\left(\bm\Omega(\bm\theta^T, \bm \rho\!)\right)\!\!|}{|\bm\Omega(\bm\theta^T, \bm \rho\!)\!|}v,\quad
w\!\rightarrow\!\h w\!=\!\frac{|T_{_{n_3}}\!\!\!\left(\bm\Omega(\bm\theta^T, \bm \rho\!)\right)\!\!|}{|\bm\Omega(\bm\theta^T, \bm \rho\!)\!|}w,
\end{align}
\end{subequations}
leave \eqref{matrix-LP} and \eqref{matrix-LP-2}invariant, respectively, where vectors $\bm\theta=(\theta^1,\dots,\theta^N)^T$, $\bm\rho=(\rho^1,\dots,\rho^N)^T$, $\bm\Omega(\bm\theta, \bm \rho^T)=(\omega(\theta^{(i)}, \rho^{(j)}))_{i,j=1,\dots,N}$, $\bm\Omega(\bm\theta^T, \bm\rho)=\bm\Omega(\bm\theta, \bm \rho^T)^T$ are $N\times N$ matrices, $\bm\Omega(\phi, \bm \rho^T)=(\omega(\phi,\rho^{(j)}))_{j=1,\dots,N}$ and $\bm\Omega(\bm \theta^T,\psi)=(\omega(\theta^{(i)},\psi))_{i=1,\dots,N}$ are $N$-row vectors, and same for the $\overline{\cdot}$ and $\b{\b\cdot}$ cases.
\end{prop}

\begin{proof}
The proof is by induction.  Let $(\theta^{1}\!, \b\theta^{1}\!, \b{\b\theta}^{1}\!)^T\!\!=\!\!(\theta[0], \b\theta[0], \b{\b\theta}[0])^T, (\theta^{2}\!,\b\theta^{2}\!, \b{\b\theta}^2\!)^T,  \dots,  (\theta^{N}\!, \b\theta^{N}\!,\b{\b\theta}^{N}\!)^T$ and $(\rho^{1}\!, \b\rho^{1}\!, \b{\b\rho}^{1}\!)^T\!\!=\!\!(\rho[0], \b\rho[0], \b{\b\rho}[0])^T, (\rho^{2}\!, \b\rho^{2}\!, \b{\b\rho}^{2})^T, \dots,  (\rho^{N}\!, \b\rho^{N}\!, \b{\b\rho}^{N}\!)^T\!$ be vector eigenfunctions of `seed' linear system \eqref{matrix-LP} and \eqref{matrix-LP-2}  respectively for the `seed' potentials $(v, w)=(v[0], w[0])$ and let $\left(\phi, \b\phi, \b{\b\phi}\right)=\left(\phi[0], \b\phi[0], \b{\b\phi}[0]\right)$, $\left(\psi, \b\psi,\b{\b\psi}\right)=\left(\psi[0], \b\psi[0], \b{\b\psi}[0]\right)$ denote arbitrary eigenfunctions.

The $N$th iteration of binary Darboux transformations is via the formulae, $N=1,  2, \dots$,
\begin{subequations}\label{binary-N-DT}
\begin{align}
&\phi[N]=\phi[N-1]-\theta[N-1]\omega(\theta[N-1], \rho[N-1])^{-1}\omega(\phi[N-1], \rho[N-1]), \label{binary-N-DT1}\\
&\b\phi[N]=\b\phi[N-1]-\b\theta[N-1]\omega(\b\theta[N-1], \b\rho[N-1])^{-1}\omega(\b\phi[N-1], \b\rho[N-1]),\label{binary-N-DT2}\\
&\b{\b\phi}[N]=\b{\b\phi}[N-1]-\b{\b\theta}[N-1]\omega(\b{\b\theta}[N-1], \b{\b\rho}[N-1])^{-1}\omega(\b{\b\phi}[N-1], \b{\b\rho}[N-1]), \label{binary-N-DT3}\\
&\psi[N]=\psi[N-1]-\rho[N-1]\omega(\theta[N-1], \rho[N-1])^{-1}\omega(\theta[N-1], \psi[N-1]), \label{binary-N-DT4} \\
&\b\psi[N]=\b\psi[N-1]-\b\rho[N-1]\omega(\b\theta[N-1], \b\rho[N-1])^{-1}\omega(\b\theta[N-1], \b\psi[N-1]), \label{binary-N-DT5}\\
&\b{\b\psi}[N]=\b{\b\psi}[N-1]-\b{\b\rho}[N-1]\omega(\b{\b\theta}[N-1], \b{\b\rho}[N-1])^{-1}\omega(\b{\b\theta}[N-1], \b{\b\psi}[N-1]), \label{binary-N-DT6}\\
&v[N]=\frac{T_{n_3}^2\!\!\left(\omega(\theta[N-1], \rho[N-1])\right)}{\omega\left(\theta[N-1], \rho[N-1]\right)}v[N-1], ~w[N]=\frac{T_{n_3}\!\!\left(\omega(\theta[N-1], \rho[N-1])\right)}{\omega\left(\theta[N-1], \rho[N-1]\right)}w[N-1], \label{potentials}
\end{align}
\end{subequations}
and where
\begin{subequations}
\begin{align}
&\theta[N]=\phi[N]|_{_{\phi\rightarrow \theta^{N+1}}}, ~~\b\theta[N]=\b\phi[N]|_{_{\b\phi \rightarrow \b\theta^{N+1}}},  , ~~\b{\b\theta}[N]=\b{\b\phi}[N]|_{_{\b{\b\phi} \rightarrow \b{\b\theta}^{N+1}}},\label{thetan}\\
&\rho[N]=\psi[N]|_{_{\psi\rightarrow \rho^{N+1}}}, ~~\b\rho[N]=\b\psi[N]|_{_{\b\psi \rightarrow \b\rho^{N+1}}}, ~~\b{\b\rho}[N]=\b{\b\psi}[N]|_{_{\b{\b\psi} \rightarrow \b{\b\rho}^{N+1}}}.
\end{align}
\end{subequations}

For $N=1$, the iterated binary Darboux transformation is the basic form given by \eqref{lmBSQ-bDT1} and \eqref{lmBSQ-bDT2}.

Suppose for $N=k$, the Proposition \ref{N-binary-lmBSQ}  is right, i.e., we have
\begin{align}
&\phi[k]\!=\!\begin{vmatrix}
\bm\Omega(\bm\theta, \bm \rho^T\!) &\!\!\! \bm\theta\\
\bm\Omega(\phi, \bm \rho^T) & \!\!\!\phi \\
\end{vmatrix}
\!|\bm\Omega(\bm\theta, \bm \rho^T\!) |^{-1}\!,
~
{\b\phi[k]}\!=\!\begin{vmatrix}
\bm\Omega(\bm{\b\theta}, \bm{\b \rho}^T\!) &\!\! \!\bm{\b\theta}\\
\bm\Omega(\b\phi, \bm {\b\rho}^T\!) &\!\! \!\b\phi \\
\end{vmatrix}
\!|\bm\Omega(\bm{\b\theta}\!, \bm {\b\rho^T}) |^{-1}\!,
~
{\b{\b\phi}[k]}\!=\!\begin{vmatrix}
\bm\Omega(\bm{\b{\b\theta}}, \bm{\b{\b \rho}}^T\!) &\!\! \!\bm{\b{\b\theta}}\\
\bm\Omega(\b{\b\phi}, \bm {\b{\b\rho}}^T\!) &\!\! \!\b{\b\phi} \\
\end{vmatrix}
\!|\bm\Omega(\bm{\b{\b\theta}}\!, \bm {\b{\b\rho}^T}) |^{-1}\!,  \label{phik}\\
&\psi[k]\!=\!\begin{vmatrix}
\bm\Omega(\bm\theta^T, \bm \rho) &\!\!\! \bm\rho\\
\bm\Omega(\bm\theta^T, \psi) &\! \!\!\psi \\
\end{vmatrix}
\!|\bm\Omega(\bm\theta^T\!\!\!, \bm \rho) |^{-1},
~
{\b\psi[k]}\!=\!\begin{vmatrix}
\bm\Omega(\bm{\b\theta}^T, \bm{\b \rho}) &\!\!\! \bm{\b\rho}\\
\bm\Omega(\bm {\b\theta}^T,\b\psi) &\!\!\! \b\psi \\
\end{vmatrix}
\!|\bm\Omega(\bm{\b\theta}^T\!\!\!, \bm {\b\rho}) |^{-1},
~
{\b{\b\psi}[k]}\!=\!\begin{vmatrix}
\bm\Omega(\bm{\b{\b\theta}}^T, \bm{\b {\b\rho}}) &\!\!\! \bm{\b{\b\rho}}\\
\bm\Omega(\bm {\b{\b\theta}}^T,\b{\b\psi}) &\!\!\! \b{\b\psi }\\
\end{vmatrix}
\!|\bm\Omega(\bm{\b{\b\theta}}^T\!\!\!, \bm {\b{\b\rho}}) |^{-1},\label{psik}\\
~
&v[k]\!=\!\frac{|T_{_{n_3}}^2\!\!\!\left(\bm\Omega(\bm\theta, \bm \rho^T\!)\right)\!\!|}{|\bm\Omega(\bm\theta, \bm \rho^T\!)\!|}v,\quad
w[k]\!=\!\frac{|T_{_{n_3}}\!\!\!\left(\bm\Omega(\bm\theta, \bm \rho^T\!)\right)\!\!|}{|\bm\Omega(\bm\theta, \bm \rho^T\!)\!|}w,
\end{align}
where the column vectors are as below
\begin{align*}
&\bm\theta=(\theta^1, \theta^2, \dots, \theta^k)^T, \quad \bm{\b\theta}=(\b\theta^1, \b\theta^2, \dots, \b\theta^k)^T, \quad \bm{\b{\b\theta}}=(\b{\b\theta}^1, \b{\b\theta}^2, \dots, \b{\b\theta}^k)^T ,\\
&\bm\rho=(\rho^1, \rho^2, \dots, \rho^k)^T, \quad \bm{\b\rho}=(\b\rho^1, \b\rho^2, \dots, \b\rho^k)^T,  \quad \bm{\b{\b\rho}}=(\b{\b\rho}^1, \b{\b\rho}^2, \dots, \b{\b\rho}^k)^T,
\end{align*}
and the $k \times k$ matrices and $1 \times k$ row vectors are as follows
\begin{align*}
&\bm\Omega(\bm\theta, \bm\rho^T)
\!=\!\left(\begin{matrix}
\omega(\theta^1, \rho^1) &\omega(\theta^1, \rho^2) & \cdots & \omega(\theta^1, \rho^k)\\
\omega(\theta^2, \rho^1) &\omega(\theta^2, \rho^2) & \cdots & \omega(\theta^2, \rho^k)\\
\vdots&\vdots & \cdots &\vdots\\
\omega(\theta^k, \rho^1) &\omega(\theta^k, \rho^2) & \cdots & \omega(\theta^k, \rho^k)\\
\end{matrix}
\right)
,~\bm\Omega(\bm\theta^T, \bm\rho)=\bm\Omega(\bm\theta, \bm\rho^T)^T, \\
&
\bm\Omega(\bm{\b\theta}, \bm{\b\rho}^T)
\!=\!\!\left(\begin{matrix}
\omega(\b\theta^1, \b\rho^1) &\omega(\b\theta^1, \b\rho^2) & \cdots & \omega(\b\theta^1,\b \rho^k)\\
\omega(\b\theta^2,\b \rho^1) &\omega(\b\theta^2, \b\rho^2) & \cdots & \omega(\b\theta^2, \b\rho^k)\\
\vdots&\vdots & \cdots &\vdots\\
\omega(\b\theta^k, \b\rho^1) &\omega(\b\theta^k, \b\rho^2) & \cdots & \omega(\b\theta^k, \b\rho^k)\\
\end{matrix}
\right),~
\bm\Omega(\bm{\b\theta}^T, \bm{\b\rho})=\bm\Omega(\bm{\b\theta}, \bm{\b\rho}^T)^T,\\
&
\bm\Omega(\bm{\b{\b\theta}}, \bm{\b{\b\rho}}^T)
\!=\!\!\left(\begin{matrix}
\omega(\b{\b\theta}^1, \b\rho^1) &\omega(\b{\b\theta}^1, \b{\b\rho}^2) & \cdots & \omega(\b{\b\theta}^1,\b{\b \rho}^k)\\
\omega(\b{\b\theta}^2,\b {\b\rho}^1) &\omega(\b{\b\theta}^2, \b{\b\rho}^2) & \cdots & \omega(\b{\b\theta}^2, \b{\b\rho}^k)\\
\vdots&\vdots & \cdots &\vdots\\
\omega(\b{\b\theta}^k, \b{\b\rho}^1) &\omega(\b{\b\theta}^k, \b{\b\rho}^2) & \cdots & \omega(\b{\b\theta}^k, \b{\b\rho}^k)\\
\end{matrix}
\right),~
\bm\Omega(\bm{\b{\b\theta}}^T, \bm{\b{\b\rho}})=\bm\Omega(\bm{\b{\b\theta}}, \bm{\b{\b\rho}}^T)^T,\\
&\bm\Omega(\phi, \bm\rho^T)
=\left(\begin{matrix}
\omega(\phi, \rho^1) &\omega(\phi, \rho^2) & \cdots & \omega(\phi, \rho^k)\\
\end{matrix}
\right)
,~
\bm\Omega(\bm\theta^T, \psi )
=\left(\begin{matrix}
\omega(\theta^1, \psi) &\omega(\theta^2, \psi) & \cdots & \omega(\theta^k, \psi)\\
\end{matrix}
\right)
,\\
&
\bm\Omega(\b\phi, \bm{\b\rho}^T)
=\left(\begin{matrix}
\omega(\b\phi, \b\rho^1) &\omega(\b\phi, \b\rho^2) & \cdots & \omega(\b\phi, \b\rho^k)\\
\end{matrix}
\right)
,~
\bm\Omega(\bm{\b\theta}^T, \b\psi )
=\left(\begin{matrix}
\omega(\b\theta^1, \b\psi) &\omega(\b\theta^2, \b\psi) & \cdots & \omega(\b\theta^k, \b\psi)\\
\end{matrix}
\right),\\
&
\bm\Omega(\b{\b\phi}, \bm{\b{\b\rho}}^T)
=\left(\begin{matrix}
\omega(\b{\b\phi}, \b{\b\rho}^1) &\omega(\b{\b\phi}, \b{\b\rho}^2) & \cdots & \omega(\b{\b\phi}, \b{\b\rho}^k)\\
\end{matrix}
\right)
,~
\bm\Omega(\bm{\b{\b\theta}}^T, \b{\b\psi} )
=\left(\begin{matrix}
\omega(\b{\b\theta}^1, \b{\b\psi}) &\omega(\b{\b\theta}^2, \b{\b\psi}) & \cdots & \omega(\b{\b\theta}^k, \b{\b\psi})\\
\end{matrix}
\right).
\end{align*}
Moerover, we have
\begin{subequations}\label{omega-formula}
\begin{align}
&\omega(\phi[k], \psi[k])
=
\begin{vmatrix}
\bm\Omega(\bm\theta, \bm\rho^T) & \bm\Omega(\bm\theta, \psi)\\
\bm\Omega(\phi, \bm\rho^T) & \omega(\phi, \psi)\\
\end{vmatrix}
\left|\bm\Omega(\bm\theta, \bm\rho^T)\right|^{-1}, \label{omega1}\\
&\omega(\b\phi[k], \b\psi[k])
=
\begin{vmatrix}
\bm\Omega(\bm{\b\theta}, \bm{\b\rho}^T) & \bm\Omega(\bm{\b\theta}, \b\psi)\\
\bm\Omega(\b\phi, \bm{\b\rho}^T) & \omega(\b\phi, \b\psi)\\
\end{vmatrix}
\left|\bm\Omega(\bm{\b\theta}, \bm{\b\rho}^T)\right|^{-1},\label{omega2}\\
&\omega(\b{\b\phi}[k], \b{\b\psi}[k])
=
\begin{vmatrix}
\bm\Omega(\bm{\b{\b\theta}}, \bm{\b{\b\rho}}^T) & \bm\Omega(\bm{\b{\b\theta}}, \b{\b\psi})\\
\bm\Omega(\b{\b\phi}, \bm{\b{\b\rho}}^T) & \omega(\b{\b\phi}, \b{\b\psi})\\
\end{vmatrix}
\left|\bm\Omega(\bm{\b{\b\theta}}, \bm{\b{\b\rho}}^T)\right|^{-1}. \label{omega3}
\end{align}
\end{subequations}
The proof of \eqref{omega-formula} is as follows.  By the definition of $\omega$ given in \eqref{lmBSQ-omega-1},  and together with  \eqref{phik} and \eqref{psik},  we have
\begin{align*}
\Delta_3\left(\omega(\phi[k], \psi[k])\right)= &~\phi[k]T_{_{n_3}}(\psi[k])\nonumber\\
= &~\Delta_3\left(\omega(\phi, \psi)\right)-T_{_{n_3}}(\bm\Omega(\bm\theta^T, \psi)\bm\Omega(\bm\theta^T, \bm\rho)^{-1})\Delta_3(\bm\Omega(\phi, \bm\rho))\\
&-\Delta_3\left(\bm\Omega(\bm\theta^T, \psi)\right)\bm\Omega(\bm\theta^T, \bm\rho)^{-1}\bm\Omega(\phi, \bm\rho)\\
&+T_{_{n_3}}(\bm\Omega(\bm\theta^T, \psi)\bm\Omega(\bm\theta^T, \bm\rho)^{-1})\Delta_3(\bm\Omega(\bm\theta^T, \bm\rho))\bm\Omega(\bm\theta^T, \bm\rho)^{-1}\bm\Omega(\phi, \bm\rho)\\
=&~\Delta_3\left(\omega(\phi, \psi)-\bm\Omega(\bm\theta^T, \psi)\bm\Omega(\bm\theta^T, \bm\rho)^{-1}\bm\Omega(\phi, \bm\rho)\right)\\
=&~\Delta_3\left(\begin{vmatrix}
\bm\Omega(\bm\theta, \bm\rho^T) & \bm\Omega(\bm\theta, \psi)\\
\bm\Omega(\phi, \bm\rho^T) & \omega(\phi, \psi)\\
\end{vmatrix}
\left|\bm\Omega(\bm\theta, \bm\rho^T)\right|^{-1}\right)
\end{align*}
which means \eqref{omega1} is right. Similarly, we can get \eqref{omega2}  and \eqref{omega3} are right as well.

Note here that we use the difference operator property for matrices as follows
\begin{align*}
\Delta_n\left(CA^{-1}B\right)=T_n(CA^{-1})\Delta_n(B)-\Delta_n(C)A^{-1}B-C_nA_n^{-1}\Delta_n(A)A^{-1}B,
\end{align*}
where $\Delta_n=T_n-1$  is the difference operator, $A=A_{N\times N}$,   $B=B_{N\times 1}$ and $C=C_{1\times N}$ are arbitrary function matrix, column vector and row vector of independent discrete variable $n$ respectively.

Next, let us prove the $(k+1)$-th step is also right. By the iterated formulae \eqref{binary-N-DT1},  we have
\begin{align}
&\phi[k+1]=\phi[k]-\theta[k]\omega(\theta[k], \rho[k])^{-1}\omega(\phi[k], \rho[k]).  \label{phik1}
\end{align}
Substitute \eqref{thetan}, \eqref{phik} and \eqref{omega1} into \eqref{phik1}, we have
\begin{align*}
\phi[k+1]&=
\begin{vmatrix}
\bm\Omega(\bm\theta, \bm\rho^T) & \bm\theta\\
\bm\Omega(\phi, \bm\rho^T) & \phi\\
\end{vmatrix}
\cdot\left|\bm\Omega(\bm\theta, \bm\rho^T)\right|^{-1}
-
\begin{vmatrix}
\bm\Omega(\bm\theta, \bm\rho^T) & \bm\theta\\
\bm\Omega(\theta^{k+1}, \bm\rho^T) & \theta^{k+1}\\
\end{vmatrix}
\cdot\left|\bm\Omega(\bm\theta, \bm\rho^T)\right|^{-1}
\cdot\\
&
\begin{vmatrix}
\bm\Omega(\bm\theta, \bm\rho^T) & \bm\Omega(\bm\theta, \rho^{k+1}) \\
\bm\Omega(\theta^{k+1}, \bm\rho^T)  & \omega(\theta^{k+1}, \rho^{k+1}) \\
\end{vmatrix}^{-1}
\cdot\left|\bm\Omega(\bm\theta, \bm\rho^T)\right|
\cdot
\begin{vmatrix}
\bm\Omega(\bm\theta, \bm\rho^T) & \bm\Omega(\bm\theta, \rho^{k+1}) \\
\bm\Omega(\phi, \bm\rho^T)  & \omega(\phi, \rho^{k+1}) \\
\end{vmatrix}
\cdot\left|\bm\Omega(\bm\theta, \bm\rho^T)\right|^{-1}\\
&=\frac{\begin{vmatrix}
\bm\Omega(\bm\theta, \bm\rho^T) & \bm\theta\\
\bm\Omega(\phi, \bm\rho^T) & \phi\\
\end{vmatrix}
\cdot
\begin{vmatrix}
\bm\Omega(\bm\theta, \bm\rho^T) & \bm\Omega(\bm\theta, \rho^{k+1}) \\
\bm\Omega(\theta^{k+1}, \bm\rho^T)  & \omega(\theta^{k+1}, \rho^{k+1}) \\
\end{vmatrix}
-
\begin{vmatrix}
\bm\Omega(\bm\theta, \bm\rho^T) & \bm\theta\\
\bm\Omega(\theta^{k+1}, \bm\rho^T) & \theta^{k+1}\\
\end{vmatrix}
\cdot
\begin{vmatrix}
\bm\Omega(\bm\theta, \bm\rho^T) & \bm\Omega(\bm\theta, \rho^{k+1}) \\
\bm\Omega(\phi, \bm\rho^T)  & \omega(\phi, \rho^{k+1}) \\
\end{vmatrix}
}{\begin{vmatrix}
\bm\Omega(\bm\theta, \bm\rho^T) & \bm\Omega(\bm\theta, \rho^{k+1}) \\
\bm\Omega(\theta^{k+1}, \bm\rho^T)  & \omega(\theta^{k+1}, \rho^{k+1}) \\
\end{vmatrix}\cdot\left|\bm\Omega(\bm\theta, \bm\rho^T)\right|}\\
&=
\begin{vmatrix}
\bm\Omega(\bm\theta, \bm\rho^T) & \bm\Omega(\bm\theta, \rho^{k+1}) & \bm\theta\\
\bm\Omega(\theta^{k+1}, \bm\rho^T) & \omega(\theta^{k+1}, \rho^{k+1}) & \theta^{k+1}\\
\bm\Omega(\phi, \bm\rho^T) & \omega(\phi, \rho^{k+1}) & \phi\\
\end{vmatrix}
\cdot
\begin{vmatrix}
\bm\Omega(\bm\theta, \bm\rho^T) & \bm\Omega(\bm\theta, \rho^{k+1}) \\
\bm\Omega(\theta^{k+1}, \bm\rho^T)  & \omega(\theta^{k+1}, \rho^{k+1}) \\
\end{vmatrix}^{-1}.
\end{align*}
Here we use the Jacobi identity.  So the equation \eqref{binary-N-DT1} is right.

Similarly, we can have
\begin{align*}
\psi[k+1]&=
\begin{vmatrix}
\bm\Omega(\bm\theta^T, \bm\rho) & \bm\Omega(\theta^{k+1}, \bm\rho) & \bm\rho\\
\bm\Omega(\bm\theta^T, \rho^{k+1}) & \omega(\theta^{k+1}, \rho^{k+1}) & \rho^{k+1}\\
\bm\Omega(\bm\theta^T, \psi) & \omega(\theta^{k+1}, \psi) & \psi\\
\end{vmatrix}
\cdot
\begin{vmatrix}
\bm\Omega(\bm\theta^T, \bm\rho) & \bm\Omega(\theta^{k+1}, \bm\rho) \\
\bm\Omega(\bm\theta^T, \rho^{k+1})  & \omega(\theta^{k+1}, \rho^{k+1}) \\
\end{vmatrix}^{-1}.
\end{align*}
It means \eqref{binary-N-DT4} is right.

Furthermore, the $\b{\cdot}$ cases \eqref{binary-N-DT2}, \eqref{binary-N-DT3} can be obtained by taking the shift operators $\lambda^{-1}T_{n_3}$ and $\lambda^{-2}T_{n_3}^2$ to act on \eqref{binary-N-DT1} , respectively. The $\b{\b{\cdot}}$  cases in \eqref{binary-N-DT5}, \eqref{binary-N-DT6}  can be obtained by taking the shift operators $\lambda^{-1}T_{n_3}^{-1}$ and $\lambda^{-2}T_{n_3}^{-2}$ to act on \eqref{binary-N-DT4} , respectively.

From \eqref{potentials}, for the potentials $\left(\h v,\h w\right)$, we get
\begin{align*}
v[k+1]=&~\frac{T_{n_3}^2\left(\omega(\theta[k], \rho[k])\right)}{\omega(\theta[k], \rho[k])}v[k]\\
=&~T_{n_3}^2\left(
\begin{vmatrix}
\bm\Omega(\bm\theta, \bm\rho^T) & \bm\Omega(\bm\theta, \rho^{k+1})\\
\bm\Omega(\theta^{k+1}, \bm\rho^T) & \omega(\theta^{k+1}, \rho^{k+1})\\
\end{vmatrix}
\cdot
\left|\bm\Omega(\bm\theta, \bm\rho^T)\right|^{-1}
\right)
\cdot
\begin{vmatrix}
\bm\Omega(\bm\theta, \bm\rho^T) & \bm\Omega(\bm\theta, \rho^{k+1})\\
\bm\Omega(\theta^{k+1}, \bm\rho^T) & \omega(\theta^{k+1}, \rho^{k+1})\\
\end{vmatrix}^{-1}
\cdot
\left|\bm\Omega(\bm\theta, \bm\rho^T)\right|\\
&
\cdot
T_{n_3}^2\left(\left|
\bm\Omega(\bm\theta, \bm\rho^T)\right|\right)\cdot\left|\bm\Omega(\bm\theta, \bm\rho^T)\right|^{-1}\cdot v \\
=&~
T_{n_3}^2\left(
\begin{vmatrix}
\bm\Omega(\bm\theta, \bm\rho^T) & \bm\Omega(\bm\theta, \rho^{k+1})\\
\bm\Omega(\theta^{k+1}, \bm\rho^T) & \omega(\theta^{k+1}, \rho^{k+1})\\
\end{vmatrix}
\right)\cdot
\begin{vmatrix}
\bm\Omega(\bm\theta, \bm\rho^T) & \bm\Omega(\bm\theta, \rho^{k+1})\\
\bm\Omega(\theta^{k+1}, \bm\rho^T) & \omega(\theta^{k+1}, \rho^{k+1})\\
\end{vmatrix}^{-1}\cdot v
\end{align*}

\begin{align*}
w[k+1]=&~\frac{T_{n_3}\left(\omega(\theta[k], \rho[k])\right)}{\omega(\theta[k], \rho[k])}w[k]\\
=&~T_{n_3}\left(
\begin{vmatrix}
\bm\Omega(\bm\theta, \bm\rho^T) & \bm\Omega(\bm\theta, \rho^{k+1})\\
\bm\Omega(\theta^{k+1}, \bm\rho^T) & \omega(\theta^{k+1}, \rho^{k+1})\\
\end{vmatrix}
\cdot
\left|\bm\Omega(\bm\theta, \bm\rho^T)\right|^{-1}
\right)
\cdot
\begin{vmatrix}
\bm\Omega(\bm\theta, \bm\rho^T) & \bm\Omega(\bm\theta, \rho^{k+1})\\
\bm\Omega(\theta^{k+1}, \bm\rho^T) & \omega(\theta^{k+1}, \rho^{k+1})\\
\end{vmatrix}^{-1}
\cdot
\left|\bm\Omega(\bm\theta, \bm\rho^T)\right|\\
&
\cdot
T_{n_3}\left(\left|
\bm\Omega(\bm\theta, \bm\rho^T)\right|\right)\cdot\left|\bm\Omega(\bm\theta, \bm\rho^T)\right|^{-1}\cdot w \\
=&~
T_{n_3}\left(
\begin{vmatrix}
\bm\Omega(\bm\theta, \bm\rho^T) & \bm\Omega(\bm\theta, \rho^{k+1})\\
\bm\Omega(\theta^{k+1}, \bm\rho^T) & \omega(\theta^{k+1}, \rho^{k+1})\\
\end{vmatrix}
\right)\cdot
\begin{vmatrix}
\bm\Omega(\bm\theta, \bm\rho^T) & \bm\Omega(\bm\theta, \rho^{k+1})\\
\bm\Omega(\theta^{k+1}, \bm\rho^T) & \omega(\theta^{k+1}, \rho^{k+1})\\
\end{vmatrix}^{-1}\cdot w.
\end{align*}
\end{proof}

Through the 3-periodic properties of the eigenfunctions $(\theta^{1}, \b\theta^{1}, \b{\b\theta}^{1})^T, (\theta^{2},\b\theta^{2}, \b{\b\theta}^{2})^T, \dots, (\theta^{N}, \b\theta^{N}, \b{\b\theta}^{N})^T$, $(\phi, \b\phi, \b{\b\phi})$, $(\rho^{1}, \b\rho^{1}, \b{\b\rho}^{1})^T, (\rho^{2},\b\rho^{2}, \b{\b\rho}^{2})^T,  \dots,  (\rho^{N}, \b\rho^{N}, \b{\b\rho}^{N})^T$ and $(\psi, \b\psi, \b{\b\psi})$, we can obtain the 3-periodic relationships for the new eigenfunctions $(\h\phi, \h{\b\phi}, \h{\b{\b\phi}})$ and $(\h\psi, \h{\b\psi}, \h{\b{\b\psi}})$ same as that shown in \eqref{3-periodic-cond-4}. Similarly, for the new potentials $(\h v, \h w)$, we have
\begin{align*}
&\h v\!=\!\frac{|T_{_{n_3}}^2\!\!\!\left(\bm\Omega(\bm\theta, \bm \rho^T\!)\right)\!\!|}{|\bm\Omega(\bm\theta, \bm \rho^T\!)\!|}v
\!=\!T_{n_3}\left(\frac{|\bm\Omega(\bm{\b{\theta}}\!, \bm {\b{\rho}}^T\!)|}{\left|T_{_{n_3}}\!\!\left(\bm\Omega(\bm{\b{\theta}}\!, \bm {\b{\rho}}^T\!)\right)\right|}\right)v
\!=\!\frac{|\bm\Omega(\bm{\b{\b\theta}}\!, \bm {\b{\b\rho}}^T\!)\!|}{\left|T_{_{n_3}}\!\!\!\left(\bm\Omega(\bm{\b{\b\theta}}\!, \bm {\b{\b\rho}}^T\!)\right)\!\right|}v,\\
&\h w\!=\!\frac{|T_{_{n_3}}\!\!\!\left(\bm\Omega(\bm\theta, \bm \rho^T\!)\right)\!\!|}{|\bm\Omega(\bm\theta, \bm \rho^T\!)\!|}w
\!=\!\frac{|\bm\Omega(\bm{\b\theta}\!, \bm {\b\rho}^T\!)\!|}{\left|T_{_{n_3}}^2\!\!\!\left(\bm\Omega(\bm{\b\theta}\!, \bm {\b\rho}^T\!)\right)\!\right|}w
\!=\!T_{n_3}\left(\frac{\left|T_{_{n_3}}\!\!\left(\bm\Omega(\bm{\b{\b\theta}}\!, \bm {\b{\b\rho}}^T\!)\right)\right|}{|\bm\Omega(\bm{\b{\b\theta}}\!, \bm {\b{\b\rho}}^T\!)|}\right)w.
\end{align*}


\section{Explicit solutions obtained by Darboux transformations and binary Darboux transformations}
Here we present explicit examples of the classes of solutions that may be obtained by means of the Darboux and binary transformations derived above.
We choose the seed solution of the discrete modified Boussinesq equation \eqref{d-p-MBSQ} as $v=v_0=1, w=w_0=1$.  With this choice, the first linear system \eqref{matrix-LP} reads
\begin{align*}
\phi_1=a_1\phi+\lambda\b\phi,\\
\b\phi_1=a_1\b\phi+\lambda\b{\b\phi},\\
\b{\b\phi}_1=a_1\b{\b\phi}+\lambda{\phi},
\end{align*}
\begin{align*}
\phi_2=a_2\phi+\lambda\b\phi,\\
\b\phi_2=a_2\b\phi+\lambda\b{\b\phi},\\
\b{\b\phi}_2=a_2\b{\b\phi}+\lambda{\phi},
\end{align*}
and the eigenfunctions are found to be, with $\omega=\frac{-1\!+\!\sqrt{3}i}{2}$, $\omega^*=\frac{-1\!-\!\sqrt{3}i}{2}$ are the solutions of $\omega^3=1$ for $\omega\neq1$,
\begin{subequations}\label{lmBSQ-seed-linear-solu-1}
\begin{align}
&\phi(n_1, n_2, n_3; \lambda)\!=\!\phi^{(1)}(n_1, n_2, n_3; \lambda)\!+\!\phi^{(2)}(n_1, n_2, n_3; \lambda)\!+\!\phi^{(3)}(n_1, n_2, n_3; \lambda),\\
&\b\phi(n_1, n_2, n_3; \lambda)\!=\!\phi^{(1)}(n_1, n_2, n_3; \lambda)\!+\!\omega\phi^{(2)}(n_1, n_2, n_3; \lambda)\!+\!\omega^*\phi^{(3)}(n_1, n_2, n_3; \lambda),\\
&\b{\b\phi}(n_1, n_2, n_3; \lambda)\!=\!\phi^{(1)}(n_1, n_2, n_3; \lambda)\!+\!\omega^{2}\!\phi^{(2)}(n_1, n_2, n_3; \lambda)\!+\!\omega^{*^{2}}\!\phi^{(3)}(n_1, n_2, n_3; \lambda),
\end{align}
\end{subequations}
where
\begin{align*}
&\phi^{(1)}(n_1, n_2, n_3; \lambda)=\lambda^{n_3}\!\!\prod_{j=1}^{2}(a_j\!+\!\lambda)^{n_j},\\
&\phi^{(2)}(n_1, n_2, n_3; \lambda)=\left(\omega\lambda\right)^{n_3}\!\!\prod_{j=1}^{2}(a_j\!+\!\omega\lambda)^{n_j},\\
&\phi^{(3)}(n_1, n_2, n_3; \lambda)=\left(\omega^*\lambda\right)^{n_3}\!\!\prod_{j=1}^{2}(a_j\!+\!\omega^*\lambda)^{n_j},
\end{align*}
which hold the 3-periodic conditions $\phi=\lambda^{-3}T_{_{n_3}}^3(\phi)$, $\b\phi=\lambda^{-1}T_{_{n_3}}(\phi)$, $\b{\b\phi}=\lambda^{-2}T_{_{n_3}}^2(\phi)$.

In a similar way the eigenfunctions of the second linear system \eqref{matrix-LP-2} are
 \begin{subequations}\label{lmBSQ-seed-linear-solu-2}
\begin{align}
&\psi(n_1, n_2, n_3; \lambda)\!=\!\psi^{(1)}(n_1, n_2, n_3; \lambda)\!+\!\psi^{(2)}(n_1, n_2, n_3; \lambda)\!+\!\psi^{(3)}(n_1, n_2, n_3; \lambda),\\
&\b\psi(n_1, n_2, n_3; \lambda)\!=\!\psi^{(1)}(n_1, n_2, n_3; \lambda)\!+\!\omega\psi^{(2)}(n_1, n_2, n_3; \lambda)\!+\!\omega^*\psi^{(3)}(n_1, n_2, n_3; \lambda),\\
&\b{\b\psi}(n_1, n_2, n_3; \lambda)\!=\!\psi^{(1)}(n_1, n_2, n_3; \lambda)\!+\!\omega^{2}\psi^{(2)}(n_1, n_2, n_3; \lambda)\!+\!\omega^{*^{2}}\psi^{(3)}(n_1, n_2, n_3; \lambda),
\end{align}
\end{subequations}
where
\begin{subequations}
\begin{align*}
&\psi^{(1)}(n_1, n_2, n_3; \lambda)=\lambda^{-n_3}\!\!\prod_{j=1}^{2}(a_j\!+\!\lambda)^{-n_j},\\
&\psi^{(2)}(n_1, n_2, n_3; \lambda)=\left(\omega\lambda\right)^{-n_3}\!\!\prod_{j=1}^{2}(a_j\!+\!\omega\lambda)^{-n_j},\\
&\psi^{(3)}(n_1, n_2, n_3; \lambda)=\left(\omega^*\lambda\right)^{-n_3}\!\!\prod_{j=1}^{2}(a_j\!+\!\omega^*\lambda)^{-n_j},
\end{align*}
\end{subequations}
which hold the 3-periodic conditions $\psi=\lambda^{-3}T_{_{n_3}}^{-3}(\psi)$, $\b\psi=\lambda^{-1}T_{_{n_3}}^{-1}(\psi)$, $\b{\b\psi}=\lambda^{-2}T_{_{n_3}}^{-2}(\psi)$.

From these eigenfunctions \eqref{lmBSQ-seed-linear-solu-1} and \eqref{lmBSQ-seed-linear-solu-2} we may integrate \eqref{lmBSQ-omega} and obtain the potentials $(\omega, \b\omega,\b{\b\omega})$
\begin{subequations}\label{lmBSQ-seed-linear-solu-omeg}
\begin{align}
\omega(\phi, \psi)&=\lambda^{-1}\left[\left(\omega-1\right)\left(\omega^{(1)}+\omega^{(4)}+\omega^{(5)}\right)+\left(\omega^*-1\right)\left(\omega^{(2)}+\omega^{(3)}+\omega^{(6)}\right)\right],\\
\omega(\b\phi,\b\psi)&=\lambda^{-1}\left[\omega\left(\omega-1\right)\left(\omega^{(1)}+\omega^{(4)}+\omega^{(5)}\right)+\omega^*\left(\omega^*-1\right)\left(\omega^{(2)}+\omega^{(3)}+\omega^{(6)}\right)\right],\\
\omega(\b{\b\phi}, \b{\b\psi})&=\lambda^{-1}\left[\omega^*(\omega-1)\left(\omega^{(1)}+\omega^{(4)}+\omega^{(5)}\right)+\omega(\omega^*-1)\left(\omega^{(2)}+\omega^{(3)}+\omega^{(6)}\right)\right],
\end{align}
\end{subequations}
where
\begin{align*}
&\omega^{(1)}(n_1,n_2,n_3;\lambda)=\omega^{n_3}\!\!\prod_{j=1}^{2}\left(\frac{a_j\!+\!\omega\lambda}{a_j\!+\!\lambda}\right)^{n_j},&
~\omega^{(2)}(n_1,n_2,n_3;\lambda)=\omega^{*^{n_3}}\!\!\prod_{j=1}^{2}\left(\frac{a_j\!+\!\omega^*\lambda}{a_j\!+\!\lambda}\right)^{n_j},\\
&\omega^{(3)}(n_1,n_2,n_3;\lambda)=\omega^{*^{n_3+1}}\!\!\prod_{j=1}^{2}\left(\frac{a_j\!+\!\lambda}{a_j\!+\!\omega\lambda}\right)^{n_j},&
~\omega^{(4)}(n_1,n_2,n_3;\lambda)=\omega^{n_3+1}\!\!\prod_{j=1}^{2}\left(\frac{a_j\!+\!\lambda}{a_j\!+\!\omega^*\lambda}\right)^{n_j},\\
&\omega^{(5)}(n_1,n_2,n_3;\lambda)=\omega^{n_3+2}\!\!\prod_{j=1}^{2}\left(\frac{a_j\!+\!\omega^*\lambda}{a_j\!+\!\omega\lambda}\right)^{n_j},&
~\omega^{(6)}(n_1,n_2,n_3;\lambda)=\omega^{*^{n_3+2}}\!\!\prod_{j=1}^{2}\left(\frac{a_j\!+\!\omega\lambda}{a_j\!+\!\omega^*\lambda}\right)^{n_j}.
\end{align*}
They hold the 3-periodic conditions $\omega(\phi, \psi)=T_{_{n_3}}^3(\omega(\phi, \psi))$, $\omega(\b\phi, \b\psi)=T_{_{n_3}}(\omega(\phi, \psi))$, $\omega(\b{\b\phi}, \b{\b\psi})=T_{_{n_3}}^2(\omega(\phi, \psi))$.

Moreover,  for $\lambda=\lambda_k$, $(v,w)=(v_0,w_0)=(1,1)$, the first linear system \eqref{matrix-LP} has eigenfunctions
\begin{subequations}\label{lmBSQ-seed-linear-solu-theta}
\begin{align}
\theta^k(n_1, n_2, n_3; \lambda_k)=\phi(n_1, n_2, n_3; \lambda)|_{\lambda=\lambda_k},\label{lmBSQ-seed-linear-solu-theta-1}\\
\b\theta^k(n_1, n_2, n_3; \lambda_k)=\b\phi(n_1, n_2, n_3; \lambda)|_{\lambda=\lambda_k},\\
\b{\b\theta}^k(n_1, n_2, n_3; \lambda_k)=\b{\b\phi}(n_1, n_2, n_3; \lambda)|_{\lambda=\lambda_k},
\end{align}
\end{subequations}
which hold $\theta^k=\lambda_k^{-3}T_{_{n_3}}^3(\theta^k)$, $\b\theta^k=\lambda_k^{-1}T_{_{n_3}}(\theta^k)$, $\b{\b\theta}^k=\lambda_k^{-2}T_{_{n_3}}^2(\theta^k)$.

Similarly, for $\lambda=\lambda_l$, $(v,w)=(v_0,w_0)=(1,1)$,  the second linear system \eqref{matrix-LP-2} has eigenfunctions
\begin{subequations}\label{lmBSQ-seed-linear-solu-rho}
\begin{gather}
\rho^l(n_1, n_2, n_3; \lambda_l)=\psi(n_1, n_2, n_3; \lambda)|_{\lambda=\lambda_l},\label{lmBSQ-seed-linear-solu-rho-1}\\
\b\rho^l(n_1, n_2, n_3; \lambda_l)=\b\psi(n_1, n_2, n_3; \lambda)|_{\lambda=\lambda_l},\\
\b{\b\rho}^l(n_1, n_2, n_3; \lambda_l)=\b{\b\psi}(n_1, n_2, n_3; \lambda)|_{\lambda=\lambda_l},
\end{gather}
\end{subequations}
which hold  $\rho^l=\lambda_l^{-3}T_{_{n_3}}^{-3}(\rho^l)$, $\b\rho^l=\lambda_l^{-1}T_{_{n_3}}^{-1}(\rho^l)$, $\b{\b\rho}^l=\lambda_l^{-2}T_{_{n_3}}^{-2}(\rho^l)$.

From these eigenfunctions \eqref{lmBSQ-seed-linear-solu-theta} and \eqref{lmBSQ-seed-linear-solu-rho} we may integrate \eqref{lmBSQ-omega} and obtain the potentials, for $\lambda_k\neq\lambda_l$,
\begin{subequations}
\begin{align}
\omega(\theta^k, \rho^l)\!=\!\lambda_l^{\!\!-1}\!\!\left(\frac{\lambda_k}{\lambda_l}\right)^{n_3}\!\!\bigg[\!\!
&\left(\frac{\lambda_k}{\lambda_l}\!-\!1\right)^{\!-1}\!\!\left(\omega^{(1)}\!+\!\omega^{(2)}+\!\omega^{(3)}\!\right)\!+\!\left(\frac{\lambda_k}{\lambda_l}*\omega\!-\!1\right)^{-1}\!\!\!\left(\omega^{(4)}\!
+\!\omega^{(7)}\!+\!\omega^{(8)}\right)\nonumber\\
&+\!\left(\frac{\lambda_k}{\lambda_l}*\omega^*-1\right)^{-1}\left(\omega^{(5)}\!+\!\omega^{(6)}\!+\!\omega^{(9)}\right)\bigg],
\label{lmBSQ-seed-linear-solu-omeg-theta-rho-1}\\
\omega(\b\theta^k, \b\rho^l)\!=\!\lambda_l^{\!\!-1}\!\!\left(\frac{\lambda_k}{\lambda_l}\right)^{n_3}\!\!\bigg[\!\!
&\left(\frac{\lambda_k}{\lambda_l}\!-\!1\right)^{\!-1}\!\!\left(\omega^{(1)}\!+\!\omega^{(2)}+\!\omega^{(3)}\!\right)\!+\!\omega*\left(\frac{\lambda_k}{\lambda_l}*\omega\!-\!1\right)^{-1}\!\!\!\left(\omega^{(4)}\!+\!\omega^{(7)}\!+\!\omega^{(8)}\right)\nonumber\\
&+\!\omega^**\left(\frac{\lambda_k}{\lambda_l}*\omega^*-1\right)^{-1}\left(\omega^{(5)}\!+\!\omega^{(6)}\!+\!\omega^{(9)}\right)\bigg],\\
\omega(\b{\b\theta}^k, \b{\b\rho}^l)\!=\!\lambda_l^{\!\!-1}\!\!\left(\frac{\lambda_k}{\lambda_l}\right)^{n_3}\!\!\bigg[\!\!
&\left(\frac{\lambda_k}{\lambda_l}\!-\!1\right)^{\!-1}\!\!\left(\omega^{(1)}\!+\!\omega^{(2)}+\!\omega^{(3)}\!\right)\!+\!\omega^**\left(\frac{\lambda_k}{\lambda_l}*\omega\!-\!1\right)^{-1}\!\!\!\left(\omega^{(4)}\!
+\!\omega^{(7)}\!+\!\omega^{(8)}\right)\nonumber\\
&+\!\omega*\left(\frac{\lambda_k}{\lambda_l}*\omega^*-1\right)^{-1}\left(\omega^{(5)}\!+\!\omega^{(6)}\!+\!\omega^{(9)}\right)\bigg],
\end{align}
\end{subequations}
where
\begin{align*}
&\omega^{(1)}(n_1,n_2,n_3;\lambda_k,\lambda_l)=\prod_{j=1}^{2}\!\!\left(\frac{a_j+\lambda_k}{a_j+\lambda_l}\right)^{n_j},
~\omega^{(2)}(n_1,n_2,n_3;\lambda_k,\lambda_l)=\omega^*\prod_{j=1}^{2}\left(\frac{a_j\!+\!\omega\lambda_k}{a_j\!+\!\omega^*\lambda_l}\right)^{n_j},\\
&\omega^{(3)}(n_1,n_2,n_3;\lambda_k,\lambda_l)=\omega\prod_{j=1}^{2}\left(\frac{a_j\!+\!\omega^*\lambda_k}{a_j\!+\!\omega\lambda_l}\right)^{n_j},
~\omega^{(4)}(n_1,n_2,n_3;\lambda_k,\lambda_l)=\omega^{n_3}\!\!\prod_{j=1}^{2}\left(\frac{a_j\!+\!\omega\lambda_k}{a_j\!+\!\lambda_l}\right)^{n_j},\\
&\omega^{(5)}(n_1,n_2,n_3;\lambda_k,\lambda_l)=\omega^{*^{n_3}}\!\!\prod_{j=1}^{2}\left(\frac{a_j\!+\!\omega^*\lambda_k}{a_j\!+\!\lambda_l}\right)^{n_j},
~\omega^{(6)}(n_1,n_2,n_3;\lambda_k,\lambda_l)=\omega^{*^{n_3+1}}\!\!\prod_{j=1}^{2}\left(\frac{a_j\!+\!\lambda_k}{a_j\!+\!\omega\lambda_l}\right)^{n_j},\\
&\omega^{(7)}(n_1,n_2,n_3;\lambda_k,\lambda_l)=\omega^{n_3+1}\!\!\prod_{j=1}^{2}\left(\frac{a_j\!+\!\lambda_k}{a_j\!+\!\omega^*\lambda_l}\right)^{n_j},
~\omega^{(8)}(n_1,n_2,n_3;\lambda_k,\lambda_l)=\omega^{n_3+2}\!\!\prod_{j=1}^{2}\left(\frac{a_j\!+\!\omega^*\lambda_k}{a_j\!+\!\omega\lambda_l}\right)^{n_j},\\
&\omega^{(9)}(n_1,n_2,n_3;\lambda_k,\lambda_l)=\omega^{*^{n_3+2}}\!\!\prod_{j=1}^{2}\left(\frac{a_j\!+\!\omega\lambda_k}{a_j\!+\!\omega^*\lambda_l}\right)^{n_j},
\end{align*}
which hold $T_{_{n_3}}(\omega(\theta^k, \rho^l))=\left(\frac{\lambda_k}{\lambda_l}\right)\omega(\b\theta^k, \b\rho^l),~~T_{_{n_3}}^2(\omega(\theta^k, \rho^l))=\left(\frac{\lambda_k}{\lambda_l}\right)^2\omega(\b{\b\theta}^k, \b{\b\rho}^l),~~T_{_{n_3}}^3(\omega(\theta^k, \rho^l))=\left(\frac{\lambda_k}{\lambda_l}\right)^3\omega(\theta^k, \rho^l)$. For $\lambda_k=\lambda_l$,  these eigenfunctions are \eqref{lmBSQ-seed-linear-solu-omeg} taking $\lambda=\lambda_k=\lambda_l$.

Given the above expressions it is straightforward to write down the following explicit solutions for the discrete modified Boussinesq equation \eqref{d-p-MBSQ} as follows
\begin{align}\label{lmBSQ-n-solu-1}
v(n_1, n_2, n_3)\!\!=\!\!\frac{T_{_{n_3}}^2\!\!\left(C_{_{[3]}}(\theta^{1},\theta^{2}, \dots, \theta^{N})\right)}{C_{_{[3]}}(\theta^{1},\theta^{2}, \dots, \theta^{N})}v_0,\\
w(n_1, n_2, n_3)\!\!=\!\!\frac{T_{_{n_3}}\!\!\left(C_{_{[3]}}(\theta^{1},\theta^{2}, \dots, \theta^{N})\right)}{C_{_{[3]}}(\theta^{1},\theta^{2}, \dots, \theta^{N})}w_0,
\end{align}
where $\theta^k=\theta^k(n_1, n_2, n_3;\lambda_k)$ is given by \eqref{lmBSQ-seed-linear-solu-theta-1}, and $\lambda_k$ are arbitrary constants;
\begin{align}\label{lmBSQ-n-solu-2}
&v(n_1, n_2, n_3)\!\!=\!\!\frac{T_{_{n_3}}^{-1}\!\!\left(C_{_{[\b 3]}}(\rho^{1},\rho^{2}, \dots, \rho^{N})\right)}{C_{_{[\b 3]}}(\rho^{1},\rho^{2}, \dots, \rho^{N})}v_0,\\
&w(n_1, n_2, n_3)\!\!=\!\!\frac{T_{_{n_3}}^{-2}\!\!\left(C_{_{[\b 3]}}(\rho^{1},\rho^{2}, \dots, \rho^{N})\right)}{C_{_{[\b 3]}}(\rho^{1},\rho^{2}, \dots, \rho^{N})}w_0,
\end{align}
where $\rho^k=\rho^k(n_1, n_2, n_3;\lambda_k)$ is given by \eqref{lmBSQ-seed-linear-solu-rho-1} and $\lambda_k$ are arbitrary constants;
\begin{align}\label{n-solu-2}
&v(n_1, n_2, n_3)=\frac{T_{_{n_3}}^2\!\left(\mathrm {det}(\omega_{k,l})\right)}{\mathrm {det}(\omega_{k,l})}v_0, ~~~~~~(k, l= 1, 2, \dots, N)\\
&w(n_1, n_2, n_3)=\frac{T_{_{n_3}}\!\left(\mathrm {det}(\omega_{k,l})\right)}{\mathrm {det}(\omega_{k,l})}w_0,~~~~~~(k, l= 1, 2, \dots, N)
\end{align}
where $\omega_{k,l}$ is given by \eqref{lmBSQ-seed-linear-solu-omeg-theta-rho-1} with  $\omega_{k,l}=\omega(\theta^k, \rho^l)$.

\section{Conclusions}
In this paper we have tried to discuss how apply the Darboux and binary Darboux transformation to the multi-component discrete integrable equations. The key to our success is a 3-periodic reduction technique on the Hirota-Miwa equation and its Lax pair. It is different from the one  in \cite{FN-2017} on the symmetry constraints, namely in \cite{FN-2017} the bilinear discrete  modified Boussinesq equation is in two-component form and ours is in three-component form. This reduction technique helps us to obtain the well-known discrete Boussinesq equation in two-component form as well as its Lax pair in $3\times3$ matrix form (corresponding to the results in \cite{FX-2017}). Moreover, the 3-periodic reduction conditions are preserved well within the Darboux and binary Darboux transformations, namely the new potentials and eigenfunctions generated by the Darboux and binary Darboux transformation still have 3-periodic property. In the process of constructing the binary Darboux transformation, we find out that there are two different linear systems (without discrete adjoint relationship introduced in \cite{N-1997}) employed in the binary Darboux transformation for multi-dimensionally consistent discrete integrable systems, which is unlike the continuous case.

Through the $N-$fold iteration of the Darboux transformations and of the binary Darboux transformations, respectively, for the two-component discrete modified Boussinesq equation, we obtained the Casorati-type solutions (expressed in eigenfunctions of the two different linear systems respectively) and Gramm-type solutions (expressed in both eigenfunctions of the two different linear systems).The  $N$-soliton solutions are found as part of those solutions for which the vacuum potentials as constants, e.g., (1,1).

Our work in the current paper is a meaningful progress based on the work in \cite{SNZ-2014, SNZ-2017}, namely the extension of our work is possible to be used to  investigate the Darboux and binary Darboux transformation for the lattice Gel'fand-Dikii hierarchy ($N$-component discrete integrable equations) as well as its corresponding modified version (namely lattice modified Gel'fand-Dikii hierarchy) \cite{N-1997}. We have succeeded in several examples, e.g., the lattice KdV-type equations \cite{SNZ-2014, SNZ-2017} and the lattice Boussinesq-type equations in the current paper. The Darboux transformation and the binary Darboux transformations are also applicable to the unmodified equation which is Miura related to the modified equation. For example, the unmodified discrete Boussinesq equation is a three-component system composed of potentials $(x, y, z)$ as follows (see \cite{FX-2017} the expression (3.19) for $i=0, 1, 2$ and $(k_1,l_1;k_2,l_2)=(0,1;0,1)$; see also in \cite{W-2018} the expressions (4.10a) and (4.11a) for $r=3$)
\begin{align}\label{3-dBSQ}
\frac{x_2-x_{12}+y_1-y}{x_1-x_{12}+y_2-y}=\frac{y_2-y_{12}+z_1-z}{y_1-y_{12}+z_2-z}=\frac{z_2-z_{12}+x_1-x}{z_1-z_{12}+x_2-x},
\end{align}
and its Lax pair
\begin{subequations}\label{matrix-LP-dBSQ}
\begin{align}
&\bm\Phi_1=\bm L\bm\Phi,\label{matrix-LP-dBSQ-1}\\
&\bm\Phi_2=\bm M\bm\Phi, \label{matrix-LP-dBSQ-2}
\end{align}
\end{subequations}
where
\begin{gather}
\bm L=\left(
\begin{array}{ccc}
 p+y-x_1& \lambda & 0 \\
0 & p+z-y_1 & \lambda \\
\lambda & 0 & p+x-z_1 \\
\end{array}
\right)
,~~~~
\bm M=\left(
\begin{array}{ccc}
 q+y-x_2& \lambda & 0 \\
0 & q+z-y_2 & \lambda \\
\lambda & 0 & q+x-z_2 \\
\end{array}
\right)
.
\end{gather}
The Miura transformation between the discrete Boussinesq equation \eqref{3-dBSQ} and the discrete modified Boussinesq equation \eqref{d-p-MBSQ} is given by
\begin{subequations}
\begin{align}
&p+y-x_1=pw_1w^{-1},\\
&p+z-y_1=pv_1v^{-1}ww_1^{-1},\\
&p+x-z_1=pvv_1^{-1},
\end{align}
\end{subequations}
and
\begin{subequations}
\begin{align}
&q+y-x_2=qw_2w^{-1},\\
&q+z-y_2=qv_2v^{-1}ww_2^{-1},\\
&q+x-z_2=qvv_2^{-1}.
\end{align}
\end{subequations}
Hence, from the Miura transformation, we can immediately build the Darboux and binary Darboux transformations of the eigenfunctions. For the potentials we need to do some adjustment based on some algebraic calculation. The results in the paper can also be generalized to the lattice Gel'fand-Dikii hierarchy for arbitrary $N$.
\section*{Acknowledgements}
The authors (YS and JXZ) dedicate this paper to Jonathan J C Nimmo. His original idea on this topic has had an enduring influence on YS's research, both directly and indirectly. YS would like to thank Jarmo Hietarinta for his useful suggestions and hospitality when the author visited Turku University. YS is supported by the National Natural Science Foundation of China (NSFC) grant (Grant Number 11501510).

\end{document}